\documentclass[aos]{imsart}

\RequirePackage{amsthm,amsmath,amsfonts,amssymb}
\RequirePackage[numbers]{natbib}
\RequirePackage[colorlinks,citecolor=blue,urlcolor=blue]{hyperref}
\RequirePackage{graphicx}
\usepackage{amsmath}
\usepackage{amssymb}
\usepackage{bm}
\usepackage{amsthm}
\usepackage{mathtools} 
\usepackage{xr}
\usepackage{enumerate}

\usepackage{array,multicol}
\usepackage[shortlabels]{enumitem}

\usepackage{lipsum}

\newcommand\blfootnote[1]{%
  \begingroup
  \renewcommand\thefootnote{}\footnote{#1}%
  \addtocounter{footnote}{-1}%
  \endgroup
}

\startlocaldefs
\theoremstyle{plain}

\newtheorem{proposition.a}{Proposition A\ignorespaces}
\newtheorem{proposition.s}{Proposition S\ignorespaces}
\newtheorem{thm}{Theorem}
\newtheorem{thm.s}{Theorem S\ignorespaces}

\newtheorem{cor.s}{Corollary S\ignorespaces}
\newtheorem{lem}{Lemma}
\newtheorem{lem.s}{Lemma S\ignorespaces}
\theoremstyle{remark}

\newtheorem{remark}{Remark}
\newtheorem{remark.s}{Remark S\ignorespaces}
\newtheorem{condition}{Condition}
\newcommand{\A}{\bm{A}}

\newcommand{\w}{\bm{w}}

\newcommand{\y}{\bm{y}}
\newcommand{\X}{\bm{X}}

\newcommand{\h}{{h}}
\newcommand{\Z}{\bm{Z}}
\newcommand{\I}{\bm{I}}

\newcommand{\bmH}{\bm{H}}

\newcommand{\bme}{\bm{\ell}}

\newcommand{\bmeps}{\bm{\epsilon}}
\newcommand{\bmbeta}{\bm{\beta}}
\newcommand{\bmalpha}{\bm{\alpha}}

\newcommand{\var}{\mbox{Var}}

\newcommand{\tr}{\mbox{tr}}
\newcommand{\bbR}{\mathbb{R}}

\newcommand{\bmSigma}{\bm{\Sigma}}

\makeatletter
\newcommand*\rel@kern[1]{\kern#1\dimexpr\macc@kerna}
\newcommand*\widebar[1]{%
  \begingroup
  \def\mathaccent##1##2{%
    \rel@kern{0.8}%
    \overline{\rel@kern{-0.8}\macc@nucleus\rel@kern{0.2}}%
    \rel@kern{-0.2}%
  }%
  \macc@depth\@ne
  \let\math@bgroup\@empty \let\math@egroup\macc@set@skewchar
  \mathsurround\z@ \frozen@everymath{\mathgroup\macc@group\relax}%
  \macc@set@skewchar\relax
  \let\mathaccentV\macc@nested@a
  \macc@nested@a\relax111{#1}%
  \endgroup
}
\makeatother

\usepackage{newfloat}
\DeclareFloatingEnvironment[name={Supplementary Fig.},fileext=lof]{suppfigure}
\DeclareFloatingEnvironment[name={Supplementary Table}]{supptable}
\usepackage{tabularx}
\usepackage{lipsum}
\def\bxzz{\color{black}}
\def\bxzzz{\color{black}}
\def\fxintercept{\color{black}}

\endlocaldefs

\begin{document}

\begin{frontmatter}
\title{High-dimensional statistical inference for linkage disequilibrium score regression and its cross-ancestry extensions}
%\title{A sample article title with some additional note\thanksref{t1}}
\runtitle{High-dimensional statistical inference for LDSC}
%\thankstext{T1}{A sample additional note to the title.}

\begin{aug}
%%%%%%%%%%%%%%%%%%%%%%%%%%%%%%%%%%%%%%%%%%%%%%%
%% Only one address is permitted per author. %%
%% Only division, organization and e-mail is %%
%% included in the address.                  %%
%% Additional information can be included in %%
\author[A]{\fnms{Fei}~\snm{Xue}\ead[label=e1]{feixue@purdue.edu}}
\and
\author[B]{\fnms{Bingxin}~\snm{Zhao}\ead[label=e2]{bxzhao@upenn.edu}}

%\orcid{0000-0000-0000-0000}
%\author[B]{\fnms{Third}~\snm{Author}\ead[label=e3]{third@somewhere.com}}
%%%%%%%%%%%%%%%%%%%%%%%%%%%%%%%%%%%%%%%%%%%%%%
%% Addresses                                %%
%%%%%%%%%%%%%%%%%%%%%%%%%%%%%%%%%%%%%%%%%%%%%%
\address[A]{Department of Statistics,
Purdue University\printead[presep={,\ }]{e1}}

\address[B]{Department of Statistics and Data Science, University of Pennsylvania\printead[presep={,\ }]{e2}}
\end{aug}

\begin{abstract}
Linkage disequilibrium score regression (LDSC) has emerged as an essential tool for genetic and genomic analyses of complex traits, utilizing high-dimensional data derived from genome-wide association studies (GWAS). LDSC computes the linkage disequilibrium (LD) scores using an external reference panel, and integrates the LD scores with only summary data from the original GWAS. 
In this paper, we investigate LDSC within a fixed-effect data integration framework, underscoring its ability to merge multi-source GWAS data and reference panels. 
In particular, we take account of the  genome-wide dependence among the high-dimensional GWAS summary statistics, along with the block-diagonal dependence pattern in estimated LD scores. 
Our analysis uncovers several key factors of both the original GWAS and reference panel datasets that determine the performance of LDSC. 
We show that it is relatively feasible for LDSC-based estimators to achieve asymptotic normality when applied to genome-wide genetic variants (e.g., in genetic variance and covariance estimation), whereas it becomes considerably challenging when we focus on a much smaller subset of genetic variants (e.g., in partitioned heritability analysis).
Moreover, by modeling the disparities in LD patterns across different populations, we {show} that LDSC can be expanded to conduct cross-ancestry analyses using data from {genetically} distinct global populations. 
%(such as European and Asian).
We validate our theoretical findings through extensive numerical evaluations using real genetic data from the UK Biobank study.
\end{abstract}



\end{frontmatter}
%%%%%%%%%%%%%%%%%%%%%%%%%%%%%%%%%%%%%%%%%%%%%%
%% Please use \tableofcontents for articles %%
%% with 50 pages and more                   %%
%%%%%%%%%%%%%%%%%%%%%%%%%%%%%%%%%%%%%%%%%%%%%%
%\tableofcontents

%%%%%%%%%%%%%%%%%%%%%%%%%%%%%%%%%%%%%%%%%%%%%%%%%%%%%%%%%%%%%%%%%%%%
%%%%%%%%%%%%%%%%%%Introduction%%%%%%%%%%%%%%%%%%%%%%%%%%%%%%%%%%%%%%
%%%%%%%%%%%%%%%%%%%%%%%%%%%%%%%%%%%%%%%%%%%%%%%%%%%%%%%%%%%%%%%%%%%%
\section{Introduction}\label{sec1}
\blfootnote{The two authors contributed equally and are listed in alphabetical order.}

Genome-wide association studies (GWAS) are designed to explore the relationship between phenotypes and genotypes by assessing differences in the allele frequency of genetic variants across individuals with varying phenotypic observations. 
Estimating and testing genetic variance and covariance has been a crucial component of most GWAS studies \citep{uffelmann2021genome}.
Genetic variance, also known as heritability, measures the extent to which genetic factors influence phenotypes and quantifies the contribution of genetics relative to environmental factors \citep{zhu2020statistical}. Genetic covariance, on the other hand, assesses genetic similarity across complex traits, revealing insights into etiological genetic pathways and the shared genetic architecture \citep{zhang2021comparison, van2019genetic}.

Linkage disequilibrium score regression (LDSC) is a {method} that can be {used} to estimate both the genetic variance of a single trait \citep{bulik2015ld} (univariate LDSC) and the genetic covariance between a pair of traits \citep{bulik2015atlas} (bivariate LDSC). 
LDSC uses marginal summary statistics from GWAS \citep{pasaniuc2017dissecting} and {\bxzz estimated linkage disequilibrium (LD) scores} from publicly available genotype reference panels \citep{10002015global}, without the need for individual-level GWAS data.
The LDSC approach leverages the {\bxzz block-diagonal LD dependence pattern} across millions of genetic variants in the human genome. Variants in the same genomic region (or LD block) can be highly correlated, while those in different regions are typically independent \citep{berisa2016approximately} (Fig.~\ref{fig1}). 

LDSC has been empirically demonstrated to be capable of handling several critical challenges in genetics, such as the polygenic nature of complex traits \citep{boyle2017expanded, timpson2018genetic} (i.e., many small but nonzero genetic effects), sample overlap between complex trait studies \citep{zhao2022genetic}, and lack of individual-level GWAS data due to data privacy concerns \citep{pasaniuc2017dissecting}. As a result, LDSC has become one of the most widely used tools in this field, with thousands of citations across a wide range of complex traits as of {\bxzzz 2024}.
LDSC has also been extended to study the partitioned heritability with functional genomic data \citep{finucane2015partitioning,gazal2017linkage,finucane2018heritability}, estimate the polygenicity of complex traits \citep{o2019extreme}, {\bxzz assess non-additive interactions \citep{Smith2022.07.21.501001},}
and 
analyze genetic {dominance} effects \citep{palmer2023analysis}. 

%%%%%%%%%%%%%%%%%%%%%Figure 1%%%%%%%%%%%%%%%%%%%%%%%%%%%
\begin{figure}[t!]
\includegraphics[page=1,width=0.75\linewidth]{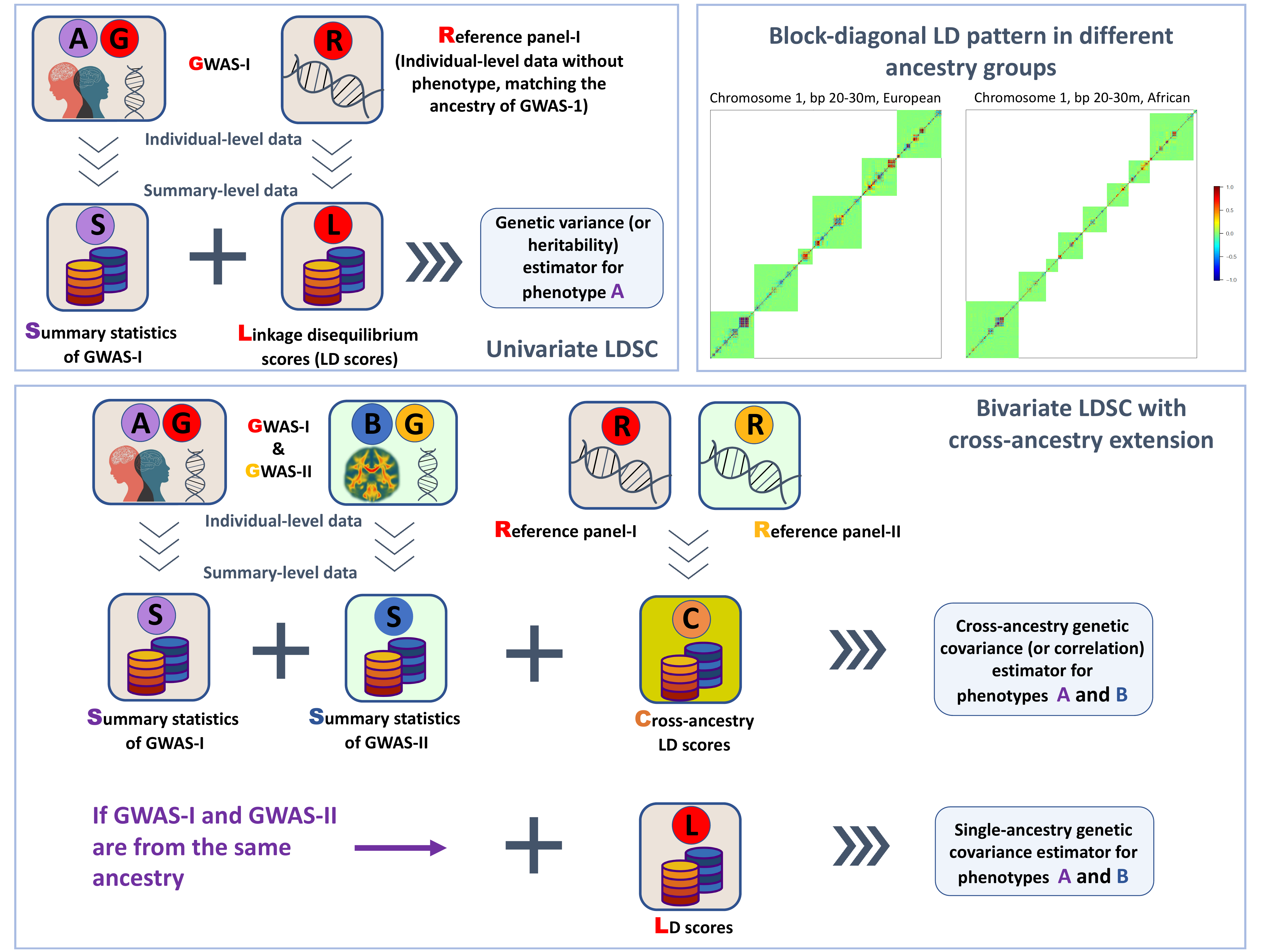}
  \caption{Illustration of LDSC estimators and {\bxzz block-diagonal} LD patterns 
  %{\color{red}\bf (block diagonal? Typo in bivariate LDSC: 'correction' should be 'correlation'.)}. 
  Univariate LDSC can estimate genetic variance (or heritability) using summary statistics from a GWAS of a particular phenotype A (e.g., depression), coupled with LD scores estimated from a reference panel. Bivariate LDSC can estimate the genetic covariance (or correlation) between phenotypes A and B (such as depression and a brain imaging trait), once again utilizing GWAS summary statistics and LD scores. In this paper, we investigate the theoretical properties of these two LDSC estimators and extend the bivariate LDSC to enable cross-ancestry applications.
  %, accounting for potential differences in LD patterns between two ancestral groups.
}
\label{fig1}
\end{figure}
%%%%%%%%%%%%%%%%%%%%%%%%%%%%%%%%%%%%%%%%%%%%%%%%%%%%%%%%

Despite the widespread use of LDSC-based estimators in genetic analyses, the understanding of their statistical properties remains limited.
Previous studies have explored the numerical performance of LDSC \citep{ni2018estimation,lee2018accuracy}, and it has been observed that LDSC may produce inconsistent or noisy estimates under specific circumstances \citep{luo2021estimating}.
Furthermore, the original LDSC papers \citep{bulik2015ld,bulik2015atlas} assume an infinitesimal {\bxzz random-effect} model, where all genetic variants have small direct contributions to the complex trait with independent and identically distributed (i.i.d.) random genetic effects. 

A recent study has broadened this assumption by examining LDSC under a misspecified random-effect model, in which a subset of genetic variants have i.i.d. random genetic effects \citep{jiang2023high}. However, it is well-known that complex traits may have complicated genetic architectures and that different variants may have varying degrees of genetic effect sizes \citep{o2021distribution}. 
For example, it is known that genetic variants collectively account for 33\% of total phenotypic variance in Alzheimer's disease. Of these, the {\it APOE} gene alone explains 6\%, other known risk genes explain 2\%, and the remaining 25\% can be jointly accounted for by all other genes, each contributing a small genetic effect \citep{ridge2013alzheimer}.
{Moreover, no previous studies have examined the impact of using a reference panel on the performance of LDSC.}
%Moreover, no previous studies have taken into account the influence of the reference panel on the estimation of LD scores. 
Therefore, it is crucial to better understand the LDSC-based estimators with a more realistic statistical modeling framework for the complex traits and LD scores. 

In this paper, we investigate the theoretical properties of LDSC within a fixed-effect data integration framework, emphasizing its competence in studying traits with complex genetic architecture and its capacity to assimilate multi-source GWAS data along with reference panels. 
Our first  major contribution lies in the explicit modeling of the genome-wide dependence across high-dimensional GWAS summary statistics, a crucial aspect that was largely overlooked in prior theoretical studies on GWAS summary data  \citep{zhao2020statistical,ye2021debiased}. 
Furthermore, we delineate the block-diagonal relatedness pattern within the estimated LD scores. 
Utilizing these results, we uncover several key factors within both the original GWAS and reference panel datasets that determine the performance of LDSC.

We show that LDSC-based estimators applied to a genome-wide scale (e.g., in genetic variance and covariance estimation \citep{bulik2015ld,bulik2015atlas}) can achieve consistency and asymptotic normality under flexible genetic architecture conditions, making them applicable to a broad range of complex traits and diseases.
However, LDSC may face additional challenges when it is used on a much smaller subset of genetic variants (e.g., in stratified heritability enrichment analysis \citep{finucane2015partitioning,gazal2017linkage,finucane2018heritability}).
Our analysis provides insights into the asymptotic behavior of LDSC-based estimators, offering valuable guidance to researchers in the field to make informed decisions when applying LDSC and interpreting results across various scenarios.

Moreover, {\bxzzz by} accounting for disparities in LD patterns across different populations, 
%(such as European and Asian), 
we formulate cross-ancestry LD scores and broaden the scope of LDSC methods to facilitate cross-ancestry genetic covariance analysis.
We also provide the theoretical properties of the cross-ancestry LDSC-based estimator.
The necessity for incorporating diverse data resources in genetic studies is steadily increasing \citep{zhou2022global}.
Cross-ancestry analysis can shed light on the generalizability and transferability of population-specific GWAS findings, hence enriching our understanding of gene-environment interactions and paving the way for enhanced genetic prediction models \citep{cai2021unified}.
The revelation that LDSC can be adapted for cross-ancestry analysis, along with the definition of cross-ancestry LD scores, could inspire the invention of many novel tools in genetics and genomics for cross-ancestry analyses.

This paper is structured as follows. 
In Section~\ref{sec2}, we introduce the datasets, models, definitions, and LDSC-based estimators. 
The theoretical results of univariate LDSC  are developed in Section~\ref{sec3}, while Section~\ref{sec4} studies bivariate LDSC with a cross-ancestry extension. 
Numerical results of extensive simulations using the UK Biobank \citep{bycroft2018uk} data are provided in Section~\ref{sec5}. 
Future topics for discussion are outlined in Section~\ref{sec6}. 
{Additional results, including technical comparisons and discussions with other existing methods and more real data analysis, as well as most}
of the technical details are provided in the Supplementary Material \citep{supplement}.

We introduce some {notations} that will be used frequently in the following {sections}.
We let $F(0,V)$ denote a generic distribution with mean $0$ and variance $V$,  
$N(0,V)$ denote a Gaussian distribution with mean $0$ and variance $V$,
and
$c,C$ represent some generic constant numbers.
We use $\tr(\cdot)$, $\|\cdot\|$, and $\|\cdot\|_k$ to denote the trace, $\ell^2$ norm, and $\ell^k$ norm of a matrix, respectively.
Let 
$H_n=o_p(h_n)$ denote $H_n/h_n\overset{p}{\to}0$,
$H_n=O_p(h_n)$ denote that $H_n/h_n$ is stochastically bounded,
$f_1(n) \lesssim f_2(n)$ denote $f_1(n)=O(f_2(n))$, 
 $f_1(n) \gtrsim f_2(n)$ denote $f_2(n)=O(f_1(n))$,
$f_1(n) \gg f_2(n)$ denote $f_2(n)/f_1(n)=o(1)$, 
$f_1(n) \ll f_2(n)$ denote $f_1(n)/f_2(n)=o(1)$,
and $f_1(n) \asymp f_2(n)$ denote that $f_1(n)$ and $f_2(n)$ have the same order of magnitude, that is, $f_1(n)=O(f_2(n))$ and $f_1(n)=O(f_2(n))$.
In addition, we let $\lambda_i(\A)$, $\lambda_{\max}(\A)$, and $\lambda_{\min}(\A)$  denote the $i$th, the largest, and the smallest eigenvalues of a matrix $\A$, respectively.

\subsection{Population descriptors}\label{sec1.1}
%for using population descriptors in genomics research 
In this paper, we define population descriptors by following the practical guidance provided by the National Academies of Science \citep{committee2023using}. Most GWAS focus on complex traits and diseases and are performed on subjects who are genetically similar to one another, often without needing to consider subjects' real ancestry or race. 
To increase sample size and combine summary statistics from different studies, current GWAS practice usually involves grouping participants under a general label, such as European (EUR). 
For example, in LDSC, the LD scores are typically estimated from the 1000 Genomes (1KG) reference panel \citep{10002015global}, which classifies individuals into five general populations: African, Admixed American, East Asian (EAS), EUR, and South Asian. 
Therefore, in this paper, we characterize GWAS participants based on their genetic similarity to one another and to a general population label in the 1KG reference panel. 
Specifically, within-ancestry LDSC analysis refers to the GWAS summary data generated from subjects genetically similar to those in one of the five general populations in the 1KG reference panel, such as 1KG-EUR. Cross-ancestry LDSC analysis involves analyzing data from two distinct groups of subjects genetically similar to two different 1KG reference panel groups, such as 1KG-EUR and 1KG-EAS. 
Genetic similarity can be typically measured by coordinates in a low-dimensional representation of genetic data, such as genetic principal components \citep{price2006principal}.

%%%%%%%%%%%%%%%%%%%%%%%%%%%%%%%%%%%%%%%%%%%%%%%%%%%%%%%%%%%%%%%%%%%%
%%%%%%%%%%%%%%%%%%LDSC-based estimators%%%%%%%%%%%%%%%%%%%%%%%%%%%%%
%%%%%%%%%%%%%%%%%%%%%%%%%%%%%%%%%%%%%%%%%%%%%%%%%%%%%%%%%%%%%%%%%%%%

\section{LDSC-based estimators}\label{sec2}
This section introduces the datasets, models, definitions of genetic terms, and LDSC-based estimators.

%%%%%%%%%%%%%%%%%%%%%%%%%%%%%%%%%%%%%%%%%%%%%%%%%%%%%%%%%%%%%%%%%%%%
%%%%%%%%%%%%%%%%%%%%%Datasets and models%%%%%%%%%%%%%%%%%%%%%%%%%%%%
%%%%%%%%%%%%%%%%%%%%%%%%%%%%%%%%%%%%%%%%%%%%%%%%%%%%%%%%%%%%%%%%%%%%
\subsection{Datasets and models}\label{sec2.1}
In this section, we describe the datasets used in our LDSC analysis.
LDSC requires two sets of input data: the first is the GWAS summary statistics, which typically include the marginal effect estimation and standard error of each genetic variant, as well as the sample size and minor allele frequency (MAF); the second set is the {\bxzz estimated} LD scores, which capture the {\bxzz block-diagonal} LD patterns among genetic variants across the genome.
The GWAS summary statistics and LD scores are typically generated from different data resources and may require careful quality control (QC) to ensure accurate and reliable results. For more information about these QC steps, please refer to \url{https://github.com/bulik/ldsc/}. 
In the models presented in the following sections, we assume that the data preprocessing steps have already been performed.

{\bxzz First, we consider two independent GWAS cohorts that generate the GWAS summary statistics for two different complex traits, respectively referred to as GWAS-I and GWAS-II:} 
\begin{itemize}
\item GWAS-I: $(\X_{\alpha},\y_{\alpha})$, with $\X_{\alpha}=(\X_{\alpha,1},\ldots,\X_{\alpha,p}) \in \bbR^{n_{\alpha}\times p}$ and $\y_\alpha \in \bbR^{n_\alpha \times 1}$; 
\item GWAS-II: $(\X_{\beta},\y_{\beta})$, with $\X_{\beta}=(\X_{\beta,1},\ldots,\X_{\beta,p}) \in \bbR^{n_{\beta}\times p}$ and $\y_{\beta} \in \bbR^{n_{\beta} \times 1}$. 
\end{itemize}
Here $\X_{\alpha}$ and $\X_{\beta}$ represent samples of $p$ genetic variants in the two GWAS. 
$\y_{\alpha}$ and $\y_{\beta}$ are two continuous complex traits studied in the two GWAS, with sample sizes $n_\alpha$ and $n_\beta$, respectively. 
Linear additive polygenic models of complex traits and the $p$ genetic variants are assumed 
\begin{flalign}
\y_{\alpha}= \X_{\alpha}\bmalpha+\bmeps_{\alpha}  \quad \text{and}  \quad
\y_{\beta}=\X_{\beta}\bmbeta+\bmeps_{\beta},  
\label{equ2.1}
\end{flalign}
where  $\bmalpha=(\alpha_1,\ldots, \alpha_{p})^{T}$ and $\bmbeta=(\beta_1,\ldots, \beta_{p})^{T}$
are unknown fixed genetic effects and $\bmeps_{\alpha}=(\epsilon_{\alpha,1},\ldots,\epsilon_{\alpha,{n_{\alpha}}})^T$ and $\bmeps_{\beta}=(\epsilon_{\beta,1},\ldots,\epsilon_{\beta,{n_{\beta}}})^T$ represent random error vectors.
{For random errors, we assume 
$\epsilon_{\alpha,i} \stackrel{i.i.d.}{\sim} N(0,\sigma^2_{\epsilon_\alpha})$ and
$\epsilon_{\beta,j} \stackrel{i.i.d.}{\sim} N(0,\sigma^2_{\epsilon_{\beta}})$, for $i=1,...,n_\alpha$ and  $j=1,...,n_{\beta}$. 

To enable cross-ancestry extensions, our main analysis assumes that the two GWAS datasets are from two different populations and thus are non-overlapping and independent. 
For the special case where the two GWAS datasets originate from the same population, our theoretical results can accommodate both partially and fully overlapping sample scenarios, as discussed in Section~\ref{sec2.3}. 
We also assume that both GWAS have the same number of genetic variants $p$, most of which are single nucleotide polymorphisms (SNPs). In practice, LDSC uses about one million genetic variants in a standard {\bxzz genome-wide} analysis, which are the variants on the HapMap3 reference panel \citep{international2010integrating}, after removing the variants in the major histocompatibility complex region. 

Next, we introduce two reference panel datasets to {\bxzz estimate} LD scores for the $p$ genetic variants:
\begin{itemize}
\item Reference panel-I: $\Z_{\alpha}=(\Z_{\alpha,1},\ldots,\Z_{\alpha,p}) \in \bbR^{n_{r\alpha}\times p}$;
\item Reference panel-II: $\Z_{\beta}=(\Z_{\beta,1},\ldots,\Z_{\beta,p}) \in \bbR^{n_{r\beta}\times p}$.
\end{itemize}
Here $\Z_{\alpha}$ and $\Z_{\beta}$ represent $n_{r\alpha}$ and $n_{r\beta}$ samples of the $p$ genetic variants in the two reference panels,
which are mutually independent and also independent of GWAS datasets $(\X_{\alpha},\y_{\alpha})$ and $(\X_{\beta},\y_{\beta})$.
The reference panel-I and reference panel-II, {typically consisting of subjects from the 1KG project \citep{10002015global}, are genetically similar to} the populations represented in GWAS-I and GWAS-II, respectively.
%required to correspond to the populations represented in GWAS-I and GWAS-II, respectively.
%Some popular reference panels {\bxzz that} are commonly used in genetics include the 1000 Genomes \citep{10002015global}, the UK10K \citep{uk10k2015uk10k}, and the TOPMed \citep{taliun2021sequencing}. 
We require the following condition for {\bxzz GWAS and reference panel} datasets. 
\begin{condition}
\label{con1}
\begin{enumerate}[(a).]
\label{con_normal_correlation}
\item We assume $\X_{\alpha}={\X_{\alpha_0}}\bmSigma_{{\alpha}}^{1/2}$, $\X_{\beta}={\X_{\beta_0}}\bmSigma_{{\beta}}^{1/2}$,
$\Z_{\alpha}={\Z_{\alpha_0}}\bmSigma_{{\alpha}}^{1/2}$, and $\Z_{\beta}={\Z_{\beta_0}}\bmSigma_{{\beta}}^{1/2}$. 
{Entries of $\X_{\alpha_0}$,  $\X_{\beta_0}$, $\Z_{\alpha_0}$ and  $\Z_{\beta_0}$ are i.i.d. Gaussian random variables with mean zero and variance one.}
The $\bmSigma_{{\alpha}}$ and $\bmSigma_{{\beta}}$ are $p\times p$ population level deterministic positive definite matrices with uniformly bounded eigenvalues. 
Specifically, we have 
$c\le \lambda_{\min}(\bmSigma_{{\alpha}}) \le \lambda_{\max}(\bmSigma_{{\alpha}})\le C$ for some positive constants $c,C$.  
The $\bmSigma_{{\beta}}$ satisfies similar conditions. 
For simplicity, we assume $\bmSigma_{\alpha_{ii}}=\bmSigma_{\beta_{ii}}=1$, $i=1,\ldots,p$.
\item 
\label{con_block}
\label{con_constant}
There exists a known $p\times p$ block-diagonal matrix structure $\mathcal{A}$ with $p_b$ diagonal blocks such that, if the $(i,j)$ element is outside {\bxzz these} $p_b$ diagonal blocks, then $\bmSigma_{{\alpha}, ij}$ and $\bmSigma_{{\beta}, ij}$ are {\bxzz zero}. The sizes of all the $p_b$ diagonal blocks are bounded by $q_b=O(1)$. 
\item\label{con_h}
We assume $\sigma_{\epsilon_{\alpha}}^2\asymp \|\bmalpha\|^2$ and $\sigma_{\epsilon_{\beta}}^2\asymp \|\bmbeta\|^2$.
\item We assume $\min(n_{\alpha},n_{\beta}, n_{r\alpha},n_{r\beta},p)\to \infty$.
\end{enumerate}
\end{condition}
In Condition~\ref{con1}~(a), $\bmSigma_{\alpha}$ and $\bmSigma_{\beta}$ are allowed to be different, which enables us to extend LDSC to cross-ancestry analysis. 
The special case $\bmSigma_{\alpha}=\bmSigma_{\beta}$ corresponds to within-ancestry applications of LDSC. 
Similar to previous high-dimensional GWAS analyses \citep{wang2022estimation,verzelen2018adaptive}, the {\bxzz normality} assumption of the datasets is used to derive our theoretical results. 
{\bxzz Our simulation results, based on real genotype data, provide evidence that this assumption of normality can be relaxed in GWAS.}

Furthermore, Condition~\ref{con1}~(b) explicitly models $\bmSigma_{\alpha}$ and $\bmSigma_{\beta}$ to be a block-diagonal matrices with {$p_b$} blocks. 
It is well-known that the LD relatedness pattern of genetic variants has a block-diagonal structure \citep{berisa2016approximately,zhao2022block}.
In LDSC algorithms, LD blocks are estimated using window sizes of $1$ centimorgan (cM) \citep{bulik2015ld}. These LD blocks typically consist of hundreds to thousands of genetic variants.
Under Condition~\ref{con1}~(b), using the block-diagonal structure $\mathcal{A}$ and within-block sample covariance estimators, we have $p\times p$ block-diagonal sample covariance matrices with $p_b$ diagonal blocks, denoted as  $\widehat{\bmSigma}_{\alpha}$ and $\widehat{\bmSigma}_{\beta}$, respectively.
In particular, if the $(i,j)$ element $\widehat{\bmSigma}_{\alpha,{ij}}$ of $\widehat{\bmSigma}_{\alpha}$ is in one of the diagonal blocks of $\mathcal{A}$, then $\widehat{\bmSigma}_{\alpha,{ij}}$ is defined as the sample covariance estimator for the $i$th and $j$th genetic variants in reference panel $\Z_{\alpha}$; otherwise, $\widehat{\bmSigma}_{\alpha,{ij}}=0$. The same definition is applied to $\widehat{\bmSigma}_{\beta}$.
By the normality assumption in Condition~\ref{con1}~(a), elements in different diagonal blocks of $\widehat{\bmSigma}_{\alpha}$ (or $\widehat{\bmSigma}_{\beta}$) are independent.

Under Condition~\ref{con1}~(c), both the genetic  (for example, $\X_{\alpha}\bmalpha$) and non-genetic (for example, $\bmeps_{\alpha}$) components have considerable contributions to the phenotype, which reflects the reality that most complex traits are influenced by both genetic and environmental factors. It is worth mentioning that there are no restrictions on the sparsity or similarity of genetic effects.
In addition, Condition~\ref{con1}~(d) assumes that all sample sizes of datasets $n_{\alpha},n_{\beta}, n_{r\alpha},n_{r\beta}$, as well as the number of genetic variants $p$, tend to infinity. We allow for these parameters to have flexible rates, which will be discussed in subsequent sections.

%%%%%%%%%%%%%%%%%%%%%%%%%%%%%%%%%%%%%%%%%%%%%%%%%%%%%%%%%%%%%%%%%%%%
%%%%%%%%%%%%%%%%%Definitions of genetic terms%%%%%%%%%%%%%%%%%%%%%%%
%%%%%%%%%%%%%%%%%%%%%%%%%%%%%%%%%%%%%%%%%%%%%%%%%%%%%%%%%%%%%%%%%%%%
\subsection{Definitions of genetic terms}\label{sec2.2}
In this section, we provide the definitions and notations for the genetic variance and covariance parameters, GWAS summary data, and LD scores.
{Following model~(\ref{equ2.1}) and Condition~\ref{con1}, the} 
genetic variances for $\y_{\alpha}$ and $\y_{\beta}$ \citep{bulik2015ld} are defined as 
\begin{flalign*}
g^2_{\alpha}=\bmalpha^T\bmalpha=\sum_{i=1}^{p}\alpha_i^2 \quad \mbox{and} \quad
g^2_{\beta}=\bmbeta^T\bmbeta=\sum_{i=1}^{p}\beta_i^2,
\end{flalign*}
which measure the aggregated genetic contributions of the $p$ genetic variants to $\y_{\alpha}$ and $\y_{\beta}$, respectively. The genetic variances are closely related to the SNP-based heritability \citep{yang2017concepts}.
For example, the SNP-based heritability of $\y_{\alpha}$ can be defined by $\h^2_{\alpha}=\bmalpha^T\bmalpha/(\bmalpha^T\bmalpha+\sigma^2_{\epsilon_\alpha})
\in (0,1)$, which is the proportion of phenotypic variance explained by genetic variance. 
For normalized phenotype with
$\var(\y_{\alpha})=1$,  we have $\h^2_{\alpha}=g^2_{\alpha}$.
The genetic covariance between $\y_{\alpha}$ and $\y_{\beta}$ \citep{bulik2015atlas} can be defined as
\begin{flalign*}
g_{\alpha\beta}=\bmalpha^T\bmbeta=\sum_{i=1}^{p}\alpha_i\beta_i,
\end{flalign*}
which quantifies the shared genetic effects between the two traits \citep{guo2019optimal,lu2017powerful,zhao2022genetic}. 
Furthermore, we define the per-variant contribution to genetic variances and covariance as 
\begin{flalign*}
\sigma_{\alpha}^2=g^2_{\alpha}/p, \quad \sigma_{\beta}^2=g^2_{\beta}/p, \quad \mbox{and} \quad
\sigma_{\alpha\beta}=g_{\alpha\beta}/p, 
\end{flalign*}
{which are target parameters in LDSC.}
As demonstrated in subsequent sections, LDSC provides separate  estimators for $\sigma_{\alpha}^2$, $\sigma_{\beta}^2$, and $\sigma_{\alpha\beta}$. Since $p$ is known in practice, estimators for $g^2_{\alpha}$, $g^2_{\beta}$, and $g_{\alpha\beta}$ can be obtained automatically.

Next, we define {\bxzz GWAS summary data as genetic effect  estimators from marginal screening \citep{uffelmann2021genome,pasaniuc2017dissecting}}
\begin{flalign*}
\widehat{a}_j=n_\alpha^{-1}\cdot \X_{\alpha, j}^T\y_\alpha \quad\mbox{and}\quad 
\widehat{b}_j=n_\beta^{-1}\cdot  \X_{\beta, j}^T\y_\beta, \quad j=1,\ldots,p.
\end{flalign*}
As $\min(n_\alpha,n_\beta)\to \infty$, we have 
\begin{flalign*}
\widehat{a}_j=\bmSigma_{\alpha,j}\bmalpha+o_p(1) \quad\mbox{and}\quad \widehat{b}_j=\bmSigma_{\beta,j}\bmbeta+o_p(1), \quad j=1,\ldots,p,
\end{flalign*}
where $\bmSigma_{\alpha,j}$ is a $1\times p$ vector corresponding to the $j$th row of $\bmSigma_{\alpha}$ and $\bmSigma_{\beta,j}$ is the $j$th row of $\bmSigma_{\beta}$. 
Let $a_j=\bmSigma_{\alpha,j}\bmalpha$ and $b_j=\bmSigma_{\beta,j}\bmbeta$ represent marginal genetic effects of the $j$th variant on $\y_{\alpha}$ and $\y_{\beta}$, respectively. Therefore, GWAS summary data are consistent estimators of these marginal genetic effects. 

Based on GWAS summary data, we introduce the following three vectors, which serve as key inputs for LDSC:
$\widehat{\bm{w}}_{a}=(\widehat{a}_1^2,\ldots,\widehat{a}_j^2,\ldots,\widehat{a}_p^2)^T$, $\widehat{\bm{w}}_{b}=(\widehat{b}_1^2,\ldots,\widehat{b}_j^2,\ldots,\widehat{b}_p^2)^T$, and $\widehat{\bm{w}}_{ab}=(\widehat{a}_1\widehat{b}_1,\ldots,\widehat{a}_j\widehat{b}_j,\ldots,\widehat{a}_p\widehat{b}_p)^T.$
To examine their statistical properties, we further define their expectations as $$\w_{a}=E(\widehat{\w}_{a})=(a_1^2, \ldots, a_p^2)^T+\{\var(\widehat{a}_1), \ldots, \var(\widehat{a}_p)\}^T,$$
$$\w_{b}=E(\widehat{\w}_{b})=(b_1^2, \ldots, b_p^2)^T+\{\var(\widehat{b}_1), \ldots, \var(\widehat{b}_p)\}^T, \quad\mbox{and} $$ 
$$\w_{ab}=E(\widehat{\w}_{ab})=(a_1b_1,\ldots,a_pb_p)^T.$$

Moreover,  we define the LD scores as 
$\bme_{a}=(\ell_{a,1},\ldots,\ell_{a,p})^T,
\bme_{b}=(\ell_{b,1},\ldots,\ell_{b,p})^T,$ and $
\bme_{ab}=(\ell_{ab,1},\ldots,\ell_{ab,p})^T$, where 
$$\ell_{a,j}=\sum_{i=1}^{p}\bmSigma_{\alpha,{ji}}^2, \quad
\ell_{b,j}=\sum_{i=1}^{p}\bmSigma_{\beta,{ji}}^2, \quad \mbox{and}\quad
\ell_{ab,j}=\sum_{i=1}^{p}\bmSigma_{\alpha,{ji}}\bmSigma_{\beta,{ji}}, \quad j=1,\ldots,p.$$
Here $\bmSigma_{\alpha,{ji}}$ and $\bmSigma_{\beta,{ji}}$ are the $(j,i)$ elements of $\bmSigma_{\alpha}$ and  $\bmSigma_{\beta}$, respectively.
The LD scores of the $j$th genetic variant quantify the aggregated relevance with other variants, either within one ancestry (for example, $\ell_{a,j}$) or between two ancestries (that is, $\ell_{ab,j}$). 
We denote $\widehat{\bme}_{a}=(\widehat{\ell}_{a, 1},\ldots,\widehat{\ell}_{a, p})^T$, 
$\widehat{\bme}_{b}=(\widehat{\ell}_{b, 1},\ldots,\widehat{\ell}_{b, p})^T$, and 
$\widehat{\bme}_{ab}=(\widehat{\ell}_{ab, 1},\ldots,\widehat{\ell}_{ab, p})^T$ to be estimators of $\bme_{a}$,
$\bme_{b}$, and $\bme_{ab}$, respectively. In practice,
they are estimated from the block-diagonal sample covariance matrices $\widehat{\bmSigma}_{\alpha}$ and $\widehat{\bmSigma}_{\beta}$. Specifically, let $\mathcal{N}(k)$ be the index set for genetic variants  in the $k$th diagonal block of $\mathcal{A}$, $k=1, \ldots, p_b$.
For $j\in \mathcal{N}(k)$, we have 
$$\widehat{\ell}_{a, j}=\sum_{i\in \mathcal{N}(k)}\widehat{\bmSigma}_{\alpha,{ji}}^2, \quad
\widehat{\ell}_{b, j}=\sum_{i\in \mathcal{N}(k)}\widehat{\bmSigma}_{\beta,{ji}}^2, \quad \mbox{and}\quad
\widehat{\ell}_{ab, j}=\sum_{i\in \mathcal{N}(k)}\widehat{\bmSigma}_{\alpha,{ji}} \widehat{\bmSigma}_{\beta,{ji}},$$ 
where $\widehat{\bmSigma}_{\alpha,{ji}}$ and $\widehat{\bmSigma}_{\beta,{ji}}$ are the $(j,i)$ elements of  $\widehat{\bmSigma}_{\alpha}$ and $\widehat{\bmSigma}_{\beta}$, respectively.
{Under the normality assumption and the block-diagonal covariance matrix structure in Condition~\ref{con1}, the $\widehat{\ell}_{a, j}$'s corresponding to genetic variants in different $\mathcal{N}(k)$'s are independent according to Lemma S1 in the Supplementary Material. 
Similar results hold for $\widehat{\ell}_{b, j}$ and 
$\widehat{\ell}_{ab, j}$ as well.} 
The following lemmas provide the statistical properties of the estimated LD scores.

\begin{lem}\label{lTl_normality_a}
Under Condition \ref{con1}, 
we have
\begin{align*}
\frac{\widehat{\bme}_{a}^T\widehat{\bme}_{a} - E(\widehat{\bme}_{a}^T\widehat{\bme}_{a})}{\rho_{l_a}}\overset{d}{\to} N(0,1)
\end{align*}
as $\min(n_{r\alpha},p)\to \infty$,
where {\bxzz $\rho_{l_a} = \{\var(\widehat{\bme}_{a}^T\widehat{\bme}_{a})\}^{1/2}\asymp (p/n_{r\alpha})^{1/2}$}. Moreover, we have 
\begin{align*}
\frac{E(\widehat{\bme}_{a}^T\widehat{\bme}_{a})
- (E\widehat{\bme}_{a})^T(E\widehat{\bme}_{a})}{p} = O(n_{r\alpha}^{-1}) \quad \mbox{and} \quad
\frac{(E\widehat{\bme}_{a})^T(E\widehat{\bme}_{a}) - \bme_{a}^T\bme_{a}}{p} = O(n_{r\alpha}^{-1}).
\end{align*}
\end{lem}
\begin{lem}\label{lTl_normality}
Under Condition \ref{con1}, 
we have
\begin{align*}
\frac{\widehat{\bme}_{ab}^T\widehat{\bme}_{ab} - E(\widehat{\bme}_{ab}^T\widehat{\bme}_{ab})}{\rho_{l_{ab}}}\overset{d}{\to} N(0,1)
\end{align*}
{as $\min(n_{r\alpha},n_{r\beta},p)\to \infty$},
where {\bxzz$\rho_{l_{ab}} = \{\var(\widehat{\bme}_{ab}^T\widehat{\bme}_{ab})\}^{1/2}\asymp (p/n_{r\alpha} + p/n_{r\beta})^{1/2}$.} Moreover, we have 
\begin{align*}
\frac{E(\widehat{\bme}_{ab}^T\widehat{\bme}_{ab})
- \bme_{ab}^T\bme_{ab}}{p}=O(n_{r\alpha}^{-1}+n_{r\beta}^{-1}).
\end{align*}
\end{lem}
Lemmas~\ref{lTl_normality_a} and~\ref{lTl_normality} quantify the accuracy of the block-diagonal estimation of LD scores, which play an important role in establishing the statistical properties of LDSC. Analogous results can be derived for $\widehat{\bme}_{b}^T\widehat{\bme}_{b}$.

{\bf \noindent Connections with random-effect notations.}
The original LDSC papers \citep{bulik2015atlas,bulik2015ld} use a random-effect model with i.i.d. genetic effects for model setups. 
Specifically, it is assumed that each variant has a nonzero small contribution to both  $\y_{\alpha}$ and $\y_{\beta}$, and the joint distribution of  $\bmalpha$ and $\bmbeta$ is given by 
\begin{flalign*}
\begin{pmatrix} 
\bmalpha\\
\bmbeta
\end{pmatrix}
\stackrel{}{\sim} N
\left (
\begin{pmatrix} 
{\bf 0}\\
{\bf 0} 
\end{pmatrix},
p^{-1} \cdot
\begin{pmatrix} 
\phi^2_{\alpha} \I_p  & \phi_{\alpha\beta} \I_p \\
\phi_{\alpha\beta} \I_p & \phi^2_{\beta} \I_p,
\end{pmatrix}
\right ).  
\end{flalign*}

It is evident that the genetic variances $g^2_{\alpha}$ and $g^2_{\beta}$ defined in our fixed-effect models are related to the terms $\phi^2_{\alpha}$ and $\phi^2_{\beta}$ in random-effect models. 
Specifically, 
$g^2_{\alpha}=\bmalpha^T\bmalpha=\sum_{i=1}^{p}\alpha_i^2$ and 
$g^2_{\beta}=\sum_{i=1}^{p}\beta_i^2$ are sample versions of $\phi^2_{\alpha}$ and $\phi^2_{\beta}$, respectively, when we consider $\alpha_i$ and $\beta_i$ as random samples.
In addition, the misspecified random-effect models \citep{jiang2016high,jiang2023high}, which consider only a subset of genetic variants with i.i.d. genetic effects, can be seen as an analog of a simplified case in our analysis: Many elements in $\bmalpha$ and $\bmbeta$ are allowed to be zero, while the remaining elements exhibit similarities.
However, our fixed-effect models offer greater flexibility by allowing the genetic effects to vary over a wide range of magnitudes. It is worth emphasizing that the investigation of the statistical properties of our flexible fixed-effect models, particularly taking into account the dependence of the high-dimensional GWAS summary data, requires the use of different technical tools than those used for simplified random-effect models.

%%%%%%%%%%%%%%%%%%%%%%%%%%%%%%%%%%%%%%%%%%%%%%%%%%%%%%%%%%%%%%%%%%%%
%%%%%%%%%%%%%%%%%LDSC-based estimators%%%%%%%%%%%%%%%%%%%%%%%%%%%%%%
%%%%%%%%%%%%%%%%%%%%%%%%%%%%%%%%%%%%%%%%%%%%%%%%%%%%%%%%%%%%%%%%%%%%
\subsection{LDSC-based estimators}\label{sec2.3}
This section presents LDSC-based estimators that are derived from the GWAS summary data and the estimated LD scores introduced in Section~\ref{sec2.2}. 
In univariate LDSC, the estimation of $\sigma_{\alpha}^2$ (or equivalently, $g^2_{\alpha}$) is based on the following equation \citep{bulik2015ld}. 
For $j=1,\ldots,p$, we have 
\begin{align}\label{equ_a_hat}
E\widehat{a}_j^2 =& E\left(n_\alpha^{-1} \X_{\alpha, j}^T{\X_{\alpha}}\bmalpha +
 n_\alpha^{-1} \X_{\alpha, j}^T \bmeps_{\alpha} \right)^2 \notag\\
  =& \sum_{i=1}^p \alpha_i^2 \bmSigma_{\alpha,ij}^2 
 + \sum_{i\neq k} \alpha_i \alpha_k \bmSigma_{\alpha,ij}\bmSigma_{\alpha,kj}
 + n_\alpha^{-1} \bmalpha^T \var (\X_{\alpha,1 j} \X_{\alpha,1 \cdot}) \bmalpha + n_\alpha^{-1} \sigma^2_{\epsilon_\alpha}\notag\\
 =& \sigma_{\alpha}^2 \cdot \sum_{i=1}^p \bmSigma_{\alpha,ij}^2 
  + \sum_{i=1}^p \alpha_i^2 \bmSigma_{\alpha,ij}^2 - \sigma_{\alpha}^2 \cdot\sum_{i=1}^p \bmSigma_{\alpha,ij}^2 
 + \sum_{i\neq k} \alpha_i \alpha_k \bmSigma_{\alpha,ij}\bmSigma_{\alpha,kj}\notag\\
 &+  n_\alpha^{-1} \bmalpha^T 
 {(\bmSigma_{\alpha} +
 \bmSigma_{\alpha, \cdot j}\bmSigma_{\alpha, \cdot j}^T)}
 %\var (\X_{\alpha,1 j} \X_{\alpha,1 \cdot}) 
 \bmalpha 
 + n_\alpha^{-1} \sigma^2_{\epsilon_\alpha}\\
  =& \sigma_{\alpha}^2 \cdot \ell_{a,j}
  + \varepsilon_{a, j},\notag
\end{align}
where $\X_{\alpha,1 j}$ denotes the $(1,j)$ element of $\X_{\alpha}$,
$\X_{\alpha, 1 \cdot}$ denotes the first row of $\X_{\alpha}$, 
%{\fxintercept $\theta_{\alpha}=n_\alpha^{-1} (\sigma^2_{\epsilon_\alpha} + g_{\alpha}^2)$},
{$\bmSigma_{\alpha, \cdot j}$ denotes the $j$th column of $\bmSigma_{\alpha}$,}
and
\begin{align*}
\varepsilon_{a, j}
% =& 
% \sum_{i=1}^p \alpha_i^2 \bmSigma_{\alpha,ij}^2 - \sigma_{\alpha}^2\cdot \sum_{i=1}^p \bmSigma_{\alpha,ij}^2 
%  + \sum_{i\neq k} \alpha_i \alpha_k \bmSigma_{\alpha,ij}\bmSigma_{\alpha,kj}
%  + n_\alpha^{-1} \bmalpha^T \var (\X_{\alpha,1 j} \X_{\alpha,1 \cdot}) \bmalpha + n_\alpha^{-1} \sigma^2_{\epsilon_\alpha}\\
 =&
    \frac{n_\alpha+1}{n_\alpha}\bmalpha^T (\bmSigma_{\alpha, \cdot j}\bmSigma_{\alpha, \cdot j}^T) \bmalpha - \sigma_{\alpha}^2\cdot \sum_{i=1}^p \bmSigma_{\alpha,ij}^2 
    + \frac{\bmalpha^T 
    \bmSigma_{\alpha}
%    \var (\X_{\alpha,1 j} \X_{\alpha,1 \cdot}) 
    \bmalpha + \sigma^2_{\epsilon_\alpha}}{n_\alpha}.
  %+ n_\alpha^{-1} \sigma^2_{\epsilon_\alpha}.
    %{- \fxintercept n_\alpha^{-1} (\sigma^2_{\epsilon_\alpha}  + g_{\alpha}^2)}.
\end{align*}
{Equation \eqref{equ_a_hat} follows from Condition \ref{con1}.}
Let $\bm{\varepsilon}_a=(\varepsilon_{a,1}, \ldots, \varepsilon_{a,p})^T$. We have 
\begin{align}\label{equ_ols_a}
\w_{a}= %{\fxintercept \theta_{\alpha} \bm{1}_p +} 
\sigma_{\alpha}^2 \cdot \bme_{a} + \bm{\varepsilon}_a.
\end{align}
{While the slope of this regression may provide an estimate for $\sigma_{\alpha}^2$, a careful evaluation under the original random-effect model \citep{bulik2015ld} suggests that the expected value of $ \varepsilon_{a, j}$ could be a non-zero constant. This indicates that it may be more appropriate to consider an intercept term in our analytical analysis for LDSC estimation of genetic variance, based on the following observation} 
% In fact, 
% under the original random-effect model \citep{bulik2015ld}, the expected value of $ \varepsilon_{a, j}$
% %we found that the average of all $ \varepsilon_{a, j}$'s 
% could be a non-zero constant. 
% This implies that it may be more appropriate to incorporate an intercept term in our analytical analysis of LDSC
% regression for genetic variance estimation, which is based on the following observation
\begin{align*}
\sigma_{\alpha}^2 = & [(\bme_{a} - \mu_{\bme_{a}} \bm{1}_p)^T(\bme_{a} - \mu_{\bme_{a}} \bm{1}_p)]^{-1}(\bme_{a} - \mu_{\bme_{a}} \bm{1}_p)^T (\w_{a} - \mu_{\w_{a}} \bm{1}_p)\\
& - [(\bme_{a} - \mu_{\bme_{a}} \bm{1}_p)^T(\bme_{a} - \mu_{\bme_{a}} \bm{1}_p)]^{-1}(\bme_{a} - \mu_{\bme_{a}} \bm{1}_p)^T(\bm{\varepsilon}_a - \mu_{\bm{\varepsilon}_a}\bm{1}_p),
\end{align*}
% }
% %$\sigma_{\alpha}^2 = (\bme_{a}^T\bme_{a})^{-1}\bme_{a}^T\w_{a}- (\bme_{a}^T\bme_{a})^{-1}\bme_{a}^T \bm{\varepsilon}_a$,
% {\fxintercept
where $\mu_{\bme_{a}}$, $\mu_{\w_{a}}$, and $\mu_{\bm{\varepsilon}_a}$ denote averages of elements in $\bme_{a}$, $\w_{a}$, and $\bm{\varepsilon}_a$, respectively. Here $\bm{1}_p$ denotes a $p$-dimensional vector with all elements being $1$'s.
More details are provided in Section S1.2 of the Supplementary Material.

In practice, LDSC estimates $\sigma_{\alpha}^2$ by regressing the squared GWAS summary data $\widehat{\w}_{a}$ on the estimated LD scores $\widehat{\bme}_{a}$. 
Therefore, the LDSC estimator of $\sigma_{\alpha}^2$ is given by
{\fxintercept
\begin{align}\label{estimator_univariate}
\widehat{\sigma}_{\alpha}^2=[(\widehat{\bme}_{a}-\widehat{\mu}_{\bme_{a}} \bm{1}_p)^T(\widehat{\bme}_{a}-\widehat{\mu}_{\bme_{a}}\bm{1}_p)]^{-1}(\widehat{\bme}_{a}-\widehat{\mu}_{\bme_{a}}\bm{1}_p)^T(\widehat{\w}_{a}-\widehat{\mu}_{\w_{a}}\bm{1}_p),
\end{align}
where $\widehat{\mu}_{\bme_{a}}$ and $\widehat{\mu}_{\w_{a}}$ denote averages of elements in $\widehat{\bme}_{a}$ and $\widehat{\w}_{a}$, respectively.}
% $$\widehat{\sigma}_{\alpha}^2=(\widehat{\bme}_{a}^T\widehat{\bme}_{a})^{-1}\widehat{\bme}_{a}^T\widehat{\w}_{a},$$
This estimator $\widehat{\sigma}_{\alpha}^2$ resembles the ordinary least squares estimator of a simple linear regression model, considering the $p$ genetic variants as individual ``samples''. 
{\bxzzz The LDSC estimator for $\sigma_{\beta}^2$, denoted as $\widehat{\sigma}_{\beta}^2$, can be defined in a similar manner.}
% Similarly, the LDSC estimator for $\sigma_{\beta}^2$ is 
% $\widehat{\sigma}_{\beta}^2=(\widehat{\bme}_{b}^T\widehat{\bme}_{b})^{-1}\widehat{\bme}_{b}^T\widehat{\w}_{b}$.
In order for $\widehat{\sigma}_{\alpha}^2$ to be a valid estimator of $\sigma_{\alpha}^2$, it is necessary that there is a certain level of orthogonality between {\bxzzz mean-centered} $\bme_{a}$ and $\bm{\varepsilon}_a$. 
Further details regarding this orthogonality will be discussed in Section \ref{sec3}.

Bivariate LDSC estimates $\sigma_{\alpha\beta}$ with two sets of GWAS summary data and the estimated LD scores
\citep{bulik2015atlas}. 
For $j=1,\ldots,p$, we have
\begin{align*}
E(\widehat{a}_j\widehat{b}_j) 
=& \sigma_{\alpha\beta}\cdot \ell_{ab,j} + \varepsilon_{ab,j},
\end{align*}
where $\varepsilon_{ab,j}= \sum_{i=1}^{p}\alpha_{i}\beta_{i}\bmSigma_{\alpha,ij}\bmSigma_{\beta,ij} - 
 \sigma_{\alpha\beta} \cdot \sum_{i=1}^{p}\bmSigma_{\alpha,ij}\bmSigma_{\beta,ij}  +
\sum_{i\neq k}\alpha_{i}\beta_{k}\bmSigma_{\alpha,ij}\bmSigma_{\beta,kj}$.
Let $\bm{\varepsilon}_{ab}=(\varepsilon_{ab,1}, \ldots, \varepsilon_{ab,p})^T$. Then we have 
\begin{align}\label{equ_ols}
\w_{ab}=\sigma_{\alpha\beta} \cdot \bme_{ab} + \bm{\varepsilon}_{ab}.
\end{align}
{\bxzzz 
It is worth noting that the expected value of $\varepsilon_{ab,j}$ is zero under the original random-effect model \cite{bulik2015atlas} for genetic covariance. Further details can be found in Section S1.2 of the Supplementary Material. 
This indicates that no intercept is explicitly required when modeling the LDSC estimator for $\sigma_{\alpha\beta}$,
}
{\fxintercept
% Note that, under the original random-effect model \cite{bulik2015atlas} for genetic covariance, the expected value of $\varepsilon_{ab,j}$ is zero (see more details in Section \ref{sub_section_intercept} of the Supplementary Material). 
% This indicates that
% no intercept is explicitly needed in modeling the LDSC estimator for $\sigma_{\alpha\beta}$, 
which is different from 
the estimator $\widehat{\sigma}_{\alpha}^2$ in (\ref{estimator_univariate}) for genetic variance. 
% ,
% while the intercept term is taken into account in
%compared with 
%the estimator $\widehat{\sigma}_{\alpha}^2$ in Equation (\ref{estimator_univariate}) for genetic variance. 
%there is no intercept involved in the estimator $\widehat{\sigma}_{\alpha\beta}$
% An intuitive %understanding of the 
% reason is that the average of $\varepsilon_{ab,j}$'s
% %elements in $\bm{\varepsilon}_{ab}$
% is close to zero when the true effects $\alpha_i$'s and $\beta_i$'s are close to those under a random-effect framework.
}{\fxintercept Since}
$\sigma_{\alpha\beta} = (\bme_{ab}^T\bme_{ab})^{-1}\bme_{ab}^T\w_{ab}- (\bme_{ab}^T\bme_{ab})^{-1}\bme_{ab}^T \bm{\varepsilon}_{ab}$,
the LDSC estimator of $\sigma_{\alpha\beta}$ {\bxzzz can be} obtained by regressing the product of GWAS summary data $\widehat{\w}_{ab}$ on the estimated LD scores $\widehat{\bme}_{ab}$, denoted as 
% {\fxintercept
% $$\widehat{\sigma}_{\alpha\beta}=[(\widehat{\bme}_{ab}-\widehat{\mu}_{\bme_{ab}} \bm{1}_p)^T(\widehat{\bme}_{ab}-\widehat{\mu}_{\bme_{ab}}\bm{1}_p)]^{-1}(\widehat{\bme}_{ab}-\widehat{\mu}_{\bme_{ab}}\bm{1}_p)^T(\widehat{\w}_{ab}-\widehat{\mu}_{\w_{ab}}\bm{1}_p),$$
% where $\widehat{\mu}_{\bme_{ab}}$ and $\widehat{\mu}_{\w_{ab}}$ denote averages of elements in $\widehat{\bme}_{ab}$ and $\widehat{\w}_{ab}$, respectively,
% %and $\bm{1}_p$ denotes a $p$-dimensional vector with all elements $1$'s.
% }
$$\widehat{\sigma}_{\alpha\beta}=(\widehat{\bme}_{ab}^T\widehat{\bme}_{ab})^{-1}\widehat{\bme}_{ab}^T\widehat{\w}_{ab}.$$

In (\ref{equ_ols}), we assume that the two GWAS cohorts used to estimate marginal genetic effects are independent. However, if there are overlapping samples between the two GWAS cohorts, the term $\bm{\varepsilon}_{ab}$ in \eqref{equ_ols} will differ, 
while the slope term $\sigma_{\alpha\beta}$ remains the same. This allows the LDSC slope to be used for estimating genetic covariance, regardless of whether the samples partially or fully overlap.  
{\bxzzz For our theoretical analysis, we may need to explicitly model the intercept in the overlapping sample case, similar to how it is handled in (\ref{estimator_univariate}) for genetic variance.}
% {\fxintercept
%  we may need to analytically model the intercept in this case
% (similar to Equation (\ref{estimator_univariate}) for genetic variance),
% }
To extend our analysis to cross-ancestry applications, we focus on the case of two independent GWAS datasets when investigating the statistical properties in later sections. In Section~\ref{sec5}, we numerically examine the effects of overlapping samples, and we provide further analysis in Section S3 of the Supplementary Material \citep{supplement} for sample overlaps.

{\fxintercept
%In addition, 
To establish theoretical results for LDSC-based estimators $\widehat{\sigma}_{\alpha}^2$ and $\widehat{\sigma}_{\alpha\beta}$, 
we introduce the following notations and condition.
}
Let 
$\bme_{ab,\max}=\max\{|\bme_{ab, i}| \text{ for } i=1,\ldots, p\}$ and
$\bme_{ab,\min}=\min\{|\bme_{ab, i}| \text{ for } i=1,\ldots, p\}$.
{We define $\widehat{\bme}_{ab, \max}$ and $\widehat{\bme}_{ab, \min}$ analogously with $\bme_{ab, i}$ replaced by $\widehat{\bme}_{ab, i}$.
}
%$\widehat{\bme}_{ab, \max} = \max\{|\widehat{\bme}_{ab, i}| \text{ for } i=1,\dots, p\}$, and 
%$\widehat{\bme}_{ab, \min} = \min\{|\widehat{\bme}_{ab, i}| \text{ for } i=1,\dots, p\}$.
{\fxintercept
We also let  $\bme_{a, \max} = \max \{|\bme_{a, i}-\mu_{\bme_a}| \text{ for } i=1, \ldots, p\}$, $\bme_{a, \min} = \min \{|\bme_{a, i}-\mu_{\bme_a}| \text{ for } i=1, \ldots, p\}$, and define $\widehat{\bme}_{a, \max}$ and $\widehat{\bme}_{a, \min}$ analogously with $\bme_{a, i}$ replaced by $\widehat{\bme}_{a, i}$.
}
%Similarly, we define $\bme_{a,\max}$, $\bme_{a,\min}$, $\widehat{\bme}_{a, \max}$, and $\widehat{\bme}_{a, \min}$. 
By Condition \ref{con1}, we  have $
\bme_{a, \max}=O(1)$ and $\bme_{ab, \max}=O(1)$.
%, and ${\bme_{a, \min}\gtrsim 1}$. 
We will need the following regularization condition for $\bme_{ab,\min}$ 
{\fxintercept
and $\bme_{a, \min}$.}

\begin{condition}\label{con_lmin_lower}\label{con_lmin_lower_2}
There exist positive constants $c$ and $C$ such that
$c\le \bme_{ab,\min}\le C$ 
{\fxintercept
and 
$\bme_{a,\min}\ge c$}.
\end{condition}

%%%%%%%%%%%%%%%%%%%%%%%%%%%%%%%%%%%%%%%%%%%%%%%%%%%%%%%%%%%%%%%%%%%%
%%%%%%%%%%%%%%%%%%%%%Univariate LDSC%%%%%%%%%%%%%%%%%%%%%%%%%%%%%%%%
%%%%%%%%%%%%%%%%%%%%%%%%%%%%%%%%%%%%%%%%%%%%%%%%%%%%%%%%%%%%%%%%%%%%
\section{Univariate LDSC}\label{sec3}
In this section, we examine the theoretical properties of $\widehat{\sigma}_{\alpha}^2$ defined in (\ref{estimator_univariate}), which is the estimator of $\sigma_{\alpha}^2$ in univariate LDSC.
{\fxintercept
First, we establish the asymptotic normality of
$(\widehat{\bme}_{a} - \widehat{\mu}_{\bme_{a}} \bm{1}_p)^T (\w_{a} - \mu_{\w_{a}} \bm{1}_p)=\widehat{\bme}_{a}^T \bmH \w_{a}$
%$\widehat{\bme}_{a}^T\w_{a}$ 
by considering the following regularity condition on $\bmH\w_{a}$, where
$\bmH=(\I_p - \bm{1}_p\bm{1}_p^T/p)^2=\I_p - \bm{1}_p\bm{1}_p^T/p$ is a centering matrix.
}

\begin{condition}\label{con_converge_rate_a}
We  assume that
{\fxintercept
$$\frac{\|\bmH\w_{a}\|_3}{\|\bmH\w_{a}\|_2}\to 0$$
}
as $p \to \infty$. Here we make the convention that {\fxintercept
$\|\bmH\w_{a}\|_3/\|\bmH\w_{a}\|_2=0$} if 
{\fxintercept
$\bmH\w_{a}=\bm{0}$}.
\end{condition}

Condition~\ref{con_converge_rate_a} can be satisfied when the number of comparable elements in {\fxintercept $\bmH\w_{a}$} diverges as $p$ goes to infinity. As the elements in $\w_{a}$ are related to marginal genetic effects, this condition aligns with the major insights emerging from GWAS findings. Specifically, it has been found that most complex traits are associated with many common variants, either through direct genetic effects or indirect effects mediated by relatedness among variants in LD \citep{visscher201710}.
Based on Conditions \ref{con1} and \ref{con_converge_rate_a}, we have the following theorem on the asymptotic behavior of 
{\fxintercept
$\widehat{\bme}_{a}^T\bmH\w_{a}$}.

\begin{thm}\label{lTw_normality_a}
Under Conditions \ref{con1} and \ref{con_converge_rate_a}, we have
{\fxintercept
\begin{align*}
\frac{(\widehat{\bme}_{a}-E\widehat{\bme}_{a})^T \bmH \w_{a}}{\rho_{a}} \overset{d}{\to} N(0,1)
\end{align*}
}
as $\min(n_{r\alpha},p)\to \infty$,
where 
{\fxintercept
$\rho_{a}^2=\var (\widehat{\bme}_{a}^T \bmH \w_{a})$.}
Moreover, we have {\fxintercept $\rho_{a} \asymp \|\bmH\w_{a}\|/\sqrt{n_{r\alpha}}$}
and 
{\fxintercept
$$(E\widehat{\bme}_{a} - \bme_a)^T \bmH \w_{a}=O\left(\frac{
%\sum_{i=1}^p |[\bmH\w_{a}]_{i}|
\|\bmH\w_{a}\|_1}{n_{r\alpha}}\right).$$}
%where $[\bmH\w_{a}]_{i}$ is the $i$th element of $\bmH\w_{a}$.
\end{thm}
Second, we establish the asymptotic normality of 
{\fxintercept
$$\widetilde{\sigma}_{\alpha}^2=[(\bme_{a} - \mu_{\bme_{a}} \bm{1}_p)^T(\bme_{a} - \mu_{\bme_{a}} \bm{1}_p)]^{-1}(\bme_{a} - \mu_{\bme_{a}} \bm{1}_p)^T (\widehat{\w}_{a} - \widehat{\mu}_{\w_{a}} \bm{1}_p),$$
}which can be considered as an oracle LDSC estimator with known LD scores. To ensure the appropriate behavior of $\widehat{\w}_{a}$, we impose the following conditions.
\begin{condition}\label{con_pn_relation_alpha}
We assume that $p=O(n_{\alpha})$.
\end{condition}
\begin{condition}\label{con_condeff_lower_a}
We assume that
$n_{\alpha}^{1/6}\|\bmalpha\| \to \infty$. 
\end{condition}

In Condition~\ref{con_pn_relation_alpha}, we assume that the number of genetic variants $p$ is proportional to or smaller than the GWAS sample size $n_\alpha$. It is worth noting that a {\bxzzz genome-wide} LDSC analysis typically involves around one million genetic variants, and it is known that the GWAS summary statistics used in LDSC cannot be generated from  GWAS with very small sample sizes, such as fewer than $5,000$ individuals \citep{zhao2022genetic}.
Condition~\ref{con_condeff_lower_a} implies that the aggregated genetic effects, represented by $\|\bmalpha\|$, cannot be too weak, and a larger GWAS sample size enables LDSC to be effective in studying traits with a lower level of genetic influence.
Under these conditions, the following theorem provides  the asymptotic normality of the oracle estimator $\widetilde{\sigma}_{\alpha}^2$, which serves as a crucial intermediate step in establishing the asymptotic normality of $\widehat{\sigma}_{\alpha}^2$.

\begin{thm}\label{thm_variance_dep}
Under {\fxintercept  Conditions 
\ref{con1}, 
\ref{con_lmin_lower},
\ref{con_pn_relation_alpha}, and 
\ref{con_condeff_lower_a},
} 
we have
\begin{equation*}
\frac{\widetilde{\sigma}_{\alpha}^2 - E\widetilde{\sigma}_{\alpha}^2}{\zeta_{\alpha} (\bme_{a})} \overset{d}{\to} N(0,1)
\end{equation*}
{as $\min(n_{\alpha},p)\to \infty$},
where $\zeta^2_{\alpha}(\bme_{a}) = \var(\widetilde{\sigma}_{\alpha}^2)$ and 
{\fxintercept
$\bm{D}_{\alpha}=\text{diag}([\bmH\bme_a]_{1}, \ldots, [\bmH\bme_a]_{p})$
with $[\bmH\bme_a]_{i}$ denoting the $i$th element in $\bmH\bme_a$.
}
%$\bm{D}_{\alpha}=\text{diag}(\sqrt{\ell_{a,1}}, \dots, \sqrt{\ell_{a,p}})$. 
Moreover, we have
\begin{align*}
\zeta^2_{\alpha}(\bme_{a})
   = & \frac{2}{n_{\alpha}^3(\bme_{a}^T \bmH \bme_{a})^2} 
   \left[ 2\sigma_{\epsilon_{\alpha}}^2 \cdot \left\{ (n_{\alpha}+2) (n_{\alpha}+3) \bmalpha^T\bm{\Sigma}_{\alpha}\bm{D}_{\alpha}\bm{\Sigma}_{\alpha} \bm{D}_{\alpha}\bm{\Sigma}_{\alpha}\bmalpha  \right.\right. \notag\\
   & %+ \left[\tr \left(\bm{\Sigma}_{\alpha}\bm{D}_{\alpha}\right)\right]^2 \bmalpha^T\bm{\Sigma}_{\alpha} \bmalpha
   \left. + (n_{\alpha} + 2) \tr \left(\bm{\Sigma}_{\alpha}\bm{D}_{\alpha} \bm{\Sigma}_{\alpha}\bm{D}_{\alpha}\right) \bmalpha^T\bm{\Sigma}_{\alpha} \bmalpha \right\}
%&\left. + (2n_{\alpha}+3) \tr \left(\bm{\Sigma}_{\alpha}\bm{D}_{\alpha}\right) \bmalpha^T\bm{\Sigma}_{\alpha} \bm{D}_{\alpha}\bm{\Sigma}_{\alpha}\bmalpha\right\}\\
 + \sigma_{\epsilon_{\alpha}}^4 \cdot (n_{\alpha}+2) \cdot \tr \left(\bm{\Sigma}_{\alpha}\bm{D}_{\alpha} \bm{\Sigma}_{\alpha}\bm{D}_{\alpha}\right) \notag\\
%+ \left[\tr \left(\bm{\Sigma}_{\alpha}\bm{D}_{\alpha}\right)\right]^2 \right\}\\
& +  (2n_{\alpha}^2+5n_{\alpha}+3)(\bmalpha^T\bm{\Sigma}_{\alpha} \bm{D}_{\alpha}\bm{\Sigma}_{\alpha}\bmalpha )^2 
+ (n_{\alpha}+2)(\bmalpha^T\bm{\Sigma}_{\alpha} \bmalpha)^2 \tr (\bm{\Sigma}_{\alpha}\bm{D}_{\alpha} \bm{\Sigma}_{\alpha}\bm{D}_{\alpha}) \notag\\
%+ (\bmalpha^T\bm{\Sigma}_{\alpha} \bmalpha)^2 \left[\tr \left(\bm{\Sigma}_{\alpha}\bm{D}_{\alpha}\right)\right]^2\\
& \left. + 2 (n_{\alpha}+2) (n_{\alpha}+3) \bmalpha^T\bm{\Sigma}_{\alpha} \bmalpha \cdot \bmalpha^T\bm{\Sigma}_{\alpha} \bm{D}_{\alpha}\bm{\Sigma}_{\alpha}\bm{D}_{\alpha}\bm{\Sigma}_{\alpha} \bmalpha \right].
%& + 2 (2n_{\alpha}+3) \bmalpha^T\bm{\Sigma}_{\alpha} \bmalpha \cdot \tr (\bm{\Sigma}_{\alpha}\bm{D}_{\alpha}) \bmalpha^T\bm{\Sigma}_{\alpha} \bm{D}_{\alpha}\bm{\Sigma}_{\alpha} \bmalpha\\
\end{align*}
\end{thm}
\begin{remark}\label{thm_variance_dep_remark}
While Theorem \ref{thm_variance_dep} is specifically derived for 
{\fxintercept mean-centered LD scores $\bmH\bme_{a}$}, it can be applied to any generic vectors as long as the absolute
values of their elements are bounded by constants from both above and below. This allows us to replace {\fxintercept $\bmH \bme_{a}$} with 
{\fxintercept $\bmH \widehat{\bme}_{a}$} when absolute
values in 
{\fxintercept
$\bmH \widehat{\bme}_{a}$} 
are bounded.
% Although Theorem \ref{thm_variance_dep} is established explicitly for LD scores  $\bme_{a}$, it holds for any generic vectors whose elements have absolute values bounded by constants from both above and below. This observation enables us to replace $\bme_{a}$ by $\widehat{\bme}_{a}$ with additional assumptions. 
\end{remark}

Although the estimated LD scores $\widehat{\bme}_{a}$ only exhibit a block-diagonal relatedness pattern, there is dependence among all elements in $\widehat{\w}_{a}$ as they are estimated based on the same vector of samples $\y_{\alpha}$ from the GWAS-I dataset. 
Handling this high-dimensional dependence poses a substantial challenge in the proof of Theorem~\ref{thm_variance_dep}. 
To address this challenge, we use techniques developed for dependent random variables to establish convergence in distribution results via total variation \citep{chatterjee2009fluctuations}. 
Additionally, {\bxzzz a more general version of} the variance estimation lemmas in \cite{dicker2014variance} {\bxzzz plays} a crucial role in deriving the explicit form of $\zeta^2_{\alpha}(\bme_{a})$.
Next, we establish the theoretical properties of the LDSC estimator $\widehat{\sigma}_{\alpha}^2$, which needs the following additional regularity conditions. 

\begin{condition}\label{sigma_con_a}
We assume that
{$\sigma_{\alpha}^2=O(1)$} and 
${\fxintercept \left|\bme_{a}^T \bmH \bm{\varepsilon}_a\right|} = o(p^{1/2})$.

\end{condition}

\begin{condition}\label{con_coneff_lower_a}
We assume that $\|\bmalpha\|^4 \gtrsim n_{\alpha} p$.
\end{condition}

\begin{condition}\label{con_alpha_4norm}  
We assume that
\begin{align*}
\frac{n_{\alpha}\|\bmalpha\|_4^4}{n_{r\alpha} \|\bmalpha\|_2^4} = o(1).
\end{align*}
\end{condition}
\begin{condition}\label{con_nr_lower_a}
We assume that $n_{r\alpha} \gg \max (\sqrt{p}, \sqrt{n_{\alpha}})$. 
\end{condition}

Condition~\ref{sigma_con_a} specifies the orthogonality between {\fxintercept $\bmH\bme_{a}$} and {\fxintercept
$\bmH\bm{\varepsilon}_a$} to ensure that the LDSC estimator $\widehat{\sigma}_{\alpha}^2$ behaves as expected and converges to its true value $\sigma_{\alpha}^2$. 
{\bxzzz 
Intuitively, in the LDSC framework, where the 
$p$ genetic variants are considered as "samples", $o(p^{1/2})$ suggests that the orthogonality level is required to be lower than the square root of the sample size. 
{\fxintercept
This orthogonality assumption essentially imposes a requirement on the true effects $\bmalpha$.}
%This can be satisfied based on previous conditions on the block-diagonal LD pattern. 
We provide more details in Section S1.3 of the Supplementary Material.}
Similar to Condition~\ref{con_condeff_lower_a}, Condition~\ref{con_coneff_lower_a} imposes constraints on the strength of aggregated genetic effects.
As indicated in (28) of the Supplementary Material, the convergence of the asymptotic normality of $\widehat{\sigma}_{\alpha}^2$ depends on the left-hand side of  Condition~\ref{con_alpha_4norm}. The sparsity of $\bmalpha$ is related to Condition~\ref{con_alpha_4norm}, since it can be easily satisfied when the number of non-zero elements in $\bmalpha$ is large and they are comparable. Thus, sparser $\bmalpha$ may result in slower convergence of the asymptotic normality of $\widehat{\sigma}_{\alpha}^2$, and LDSC is more effective for traits with denser genetic signals.
Additionally, Conditions~\ref{con_alpha_4norm} and~\ref{con_nr_lower_a} indicate that the reference panels are allowed to have a much smaller sample size than that of GWAS data. The following theorem characterizes the asymptotic normality of $\widehat{\sigma}_{\alpha}^2$. 

\begin{thm}\label{thm_clt_var}
Under 
{\fxintercept
Conditions 
\ref{con1}--\ref{con_nr_lower_a},
}
 we have
\begin{align*}
\frac{\widehat{\sigma}_{\alpha}^2 - \sigma_{\alpha}^2}{\zeta_{\alpha} (\bme_{a})}  \overset{d}{\to} N(0,1)
\end{align*}
and $\zeta_{\alpha} (\bme_{a})\to 0$
as $\min(n_{\alpha},n_{r\alpha},p)\to \infty$. 
\end{thm}

%{\bf \noindent .} 
\begin{proof}[Proof sketch of Theorem~\ref{thm_clt_var}]
We rewrite the standardized $\widehat{\sigma}_{\alpha}^2$ as follows:
{\fxintercept
\begin{align*}
\frac{\widehat{\sigma}_{\alpha}^2 - \sigma_{\alpha}^2}{\zeta_{\alpha} (\bme_{a})} 
=& \frac{(\widehat{\bme}_{a}^T \bmH \widehat{\bme}_{a})^{-1}\widehat{\bme}_{a}^T \bmH \widehat{\w}_{a} -  (\bme_{a}^T \bmH \bme_{a})^{-1}\bme_{a}^T \bmH \w_{a}}{\zeta_{\alpha} (\bme_{a})}  
+ \frac{(\bme_{a}^T \bmH \bme_{a})^{-1} \bme_{a}^T \bmH \bm{\varepsilon}_a}{\zeta_{\alpha} (\bme_{a})} \\
=& \frac{(\widehat{\bme}_{a}^T \bmH \widehat{\bme}_{a})^{-1} \widehat{\bme}_{a}^T \bmH (\widehat{\w}_{a} -\w_{a})}{\zeta_{\alpha} (\bme_{a})}   
+ \frac{\{(\widehat{\bme}_{a}^T \bmH \widehat{\bme}_{a})^{-1}
\widehat{\bme}_{a}^T -  (\bme_{a}^T \bmH \bme_{a})^{-1}\bme_{a}^T\} \bmH \w_{a}}{\zeta_{\alpha} (\bme_{a})}\\
& + \frac{(\bme_{a}^T \bmH \bme_{a})^{-1}\bme_{a}^T \bmH \bm{\varepsilon}_a}{\zeta_{\alpha} (\bme_{a})}.
\end{align*}
% \begin{align*}
% \frac{\widehat{\sigma}_{\alpha}^2 - \sigma_{\alpha}^2}{\zeta_{\alpha} (\bme_{a})} 
% =& \frac{(\widehat{\bme}_{a}^T\widehat{\bme}_{a})^{-1}\widehat{\bme}_{a}^T\widehat{\w}_{a} -  (\bme_{a}^T\bme_{a})^{-1}\bme_{a}^T\w_{a}}{\zeta_{\alpha} (\bme_{a})}  
% + \frac{(\bme_{a}^T\bme_{a})^{-1}\bme_{a}^T \bm{\varepsilon}_a}{\zeta_{\alpha} (\bme_{a})} \\
% =& \frac{(\widehat{\bme}_{a}^T\widehat{\bme}_{a})^{-1}\widehat{\bme}_{a}^T (\widehat{\w}_{a} -\w_{a})}{\zeta_{\alpha} (\bme_{a})}   
% + \frac{\{(\widehat{\bme}_{a}^T\widehat{\bme}_{a})^{-1}\widehat{\bme}_{a}^T\ -  (\bme_{a}^T\bme_{a})^{-1}\bme_{a}^T\} \w_{a}}{\zeta_{\alpha} (\bme_{a})} + \frac{(\bme_{a}^T\bme_{a})^{-1}\bme_{a}^T \bm{\varepsilon}_a}{\zeta_{\alpha} (\bme_{a})}. 
% \end{align*}
}Then it suffices to show that 
{\fxintercept
\begin{align}\label{step1_a}
\frac{(\bme_{a}^T \bmH \bme_{a})^{-1}\bme_{a}^T \bmH \bm{\varepsilon}_a}{\zeta_{\alpha} (\bme_{a})} \to 0,
\end{align}
\begin{align}\label{step2_a}
\frac{\{(\widehat{\bme}_{a}^T \bmH \widehat{\bme}_{a})^{-1}
\widehat{\bme}_{a}^T -  (\bme_{a}^T \bmH \bme_{a})^{-1}\bme_{a}^T\} \bmH \w_{a}}{\zeta_{\alpha} (\bme_{a})} \overset{p}{\to} 0,
\end{align}
and 
\begin{align}\label{step3_a}
\frac{(\widehat{\bme}_{a}^T \bmH \widehat{\bme}_{a})^{-1} \widehat{\bme}_{a}^T \bmH (\widehat{\w}_{a} -\w_{a})}{\zeta_{\alpha} (\bme_{a})}   \overset{d}{\to} N(0,1).
\end{align}
}

It is important to note that the lower and upper bounds of $\zeta_{\alpha} (\bme_{a})$ are both increasing functions of the aggregated genetic effects $\|\bmalpha\|$. Consequently, in order to derive \eqref{step1_a}, we require an upper bound on
{\fxintercept
$|\bme_{a}^T \bmH \bm{\varepsilon}_a|$}
and a lower bound on $\|\bmalpha\|$, which can be obtained from Conditions \ref{sigma_con_a} and \ref{con_coneff_lower_a}. Furthermore, the convergence of $\zeta_{\alpha} (\bme_{a})$ to $0$ can be deduced from Condition \ref{sigma_con_a}, which restricts the upper bound of $\|\bmalpha\|$. The application of Theorems \ref{lTw_normality_a} and \ref{thm_variance_dep} allows us to derive (\ref{step2_a}) and (\ref{step3_a}), respectively. Here we show $\zeta_{\alpha} (\bme_{a})\to 0$ to ensure the convergence of $\widehat{\sigma}_{\alpha}^2$ to $\sigma_{\alpha}^2$. Detailed proof of Theorem~\ref{thm_clt_var} can be found in Section S6.3 of the Supplementary Material \citep{supplement}.
\end{proof}

One of the key conditions of LDSC is the block-diagonal LD pattern in high-dimensional GWAS data. 
On the one hand, the relatedness of genetic variants establishes the linear trend between LD scores and the marginal genetic effects, while the slope between them happens to be {\bxzz $\sigma_{\alpha}^2$}. This enables the estimation of {\bxzz $\sigma_{\alpha}^2$} from the regression between the $p$ pairs of {\bxzz estimated LD scores and the GWAS summary data.} 
On the other hand, according to the block-diagonal LD pattern, only nearby genetic variants can be highly correlated, guaranteeing reliable estimation of LD scores even when the dimensionality of GWAS data is high.

More importantly, when establishing the asymptotic normality of $\widehat{\sigma}_{\alpha}^2$, we take into account the genome-wide dependence structure among GWAS summary statistics. This differs from previous studies on high-dimensional inference in GWAS data \citep{zhao2020statistical,ye2021debiased}, which often assume independence among their genetic effect estimates.

In addition to the typical LDSC-based estimators applied on a genome-wide scale, our results also provide insights when LDSC is used on a smaller subset of genetic variants, such as in stratified heritability enrichment analysis \citep{finucane2015partitioning,gazal2017linkage,finucane2018heritability}. In such LDSC applications, the number of involved genetic variants $p$ is typically much smaller compared to LDSC on a genome-wide scale \citep{bulik2015ld,bulik2015atlas}. Consequently, satisfying our conditions, especially Conditions~\ref{sigma_con_a} and \ref{con_coneff_lower_a}, becomes more challenging, and the convergence of the asymptotic normality may be slower or may not hold.
{\bxzzz We provide more discussions with real functional annotation data in the Discussion Section.}

We have explicitly quantified the asymptotic variance of $\widehat{\sigma}_{\alpha}^2$ in Theorem~\ref{thm_clt_var}, uncovering how the model parameters and dataset characteristics determine the uncertainty of estimation.
In practice, as most of these factors are unknown, LDSC uses the block jackknife to estimate variance and perform resampling-based statistical inference \citep{bulik2015atlas,bulik2015ld}. 
% {\color{red}
% %In Section~\ref{sec5}, we will evaluate the influences of these factors on LDSC in simulations using real GWAS data.
% }

Furthermore, we establish the consistency of $\widehat{\sigma}_{\alpha}^2$ 
{\fxintercept with weaker conditions in Section S2 of the Supplementary Material.}

\section{Bivariate LDSC}\label{sec4}
In this section, we investigate the theoretical properties of $\widehat{\sigma}_{\alpha\beta}=(\widehat{\bme}_{ab}^T\widehat{\bme}_{ab})^{-1}\widehat{\bme}_{ab}^T\widehat{\w}_{ab}$, 
which is the estimator of $\sigma_{\alpha\beta}^2$ in bivariate LDSC. 
To extend LDSC to cross-ancestry applications, we {\bxzzz use} cross-ancestry LD scores $\bme_{ab}$ instead of within-ancestry LD scores $\bme_{a}$. 
Similar to Section \ref{sec3}, we begin by establishing the asymptotic normality of $\widehat{\bme}_{ab}^T\w_{ab}$.
We need the following regularity condition on $\w_{ab}$.

\begin{condition}\label{con_converge_rate}
We  assume that 
$$\frac{\|\w_{ab}\|_3}{\|\w_{ab}\|_2}\to 0$$
as $p \to \infty$. Here we make the convention that $\|\w_{ab}\|_3/\|\w_{ab}\|_2=0$ if $\w_{ab}=\bm{0}$.
\end{condition}
% for $\w_{a}$
Similar to Condition~\ref{con_converge_rate_a}, Condition~\ref{con_converge_rate} can easily be satisfied for two polygenic traits when there is a diverging number (as $p$ goes to infinity) of genetic variants with nonzero comparable marginal effects, collectively influencing both traits \citep{timpson2018genetic}.
%Similar to Condition~\ref{con_converge_rate_a} for $\w_{a}$, Condition~\ref{con_converge_rate} can be easily satisfied for two polygenic traits that have a diverging number of genetic variants with nonzero marginal effects, collectively influencing both traits \citep{timpson2018genetic}.
Based on Conditions \ref{con1} and \ref{con_converge_rate}, we have the following theorem. 
\begin{thm}\label{lTw_normality}
Under Conditions \ref{con1} and \ref{con_converge_rate}, we have
\begin{align*}
\frac{(\widehat{\bme}_{ab}-E\widehat{\bme}_{ab})^T\w_{ab}}{\rho_{ab}} \overset{d}{\to} N(0,1)
\end{align*}
{as $\min(n_{r\alpha},n_{r\beta},p)\to \infty$},
where 
\begin{align*}
\rho_{ab}^2=\var (\widehat{\bme}_{ab}^T\w_{ab})
=&\sum_{m=1}^{p_b} \sum_{j, k\in \mathcal{N}(m)} \w_{ab, j} \w_{ab, k} \left[n_{r\alpha}^{-1}n_{r\beta}^{-1} \cdot \{ (\bmSigma_{\alpha} \bmSigma_{\beta})_{kk} (\bmSigma_{\alpha} \bmSigma_{\beta})_{jj} \right.\\
& + \bmSigma_{\alpha, kj} \bmSigma_{\beta, kj} \tr (\bmSigma_{\alpha, \mathcal{N}(m) \mathcal{N}(m)} \bmSigma_{\beta, \mathcal{N}(m) \mathcal{N}(m)})  \} \\
&  + \frac{1 + n_{r\beta}}{n_{r\alpha}n_{r\beta}} \cdot\bmSigma_{\alpha, kj}(\bmSigma_{\beta}\bmSigma_{\alpha} \bmSigma_{\beta})_{kj}
+ \frac{1 + n_{r\alpha}}{n_{r\alpha}n_{r\beta}} \cdot\bmSigma_{\beta, kj}(\bmSigma_{\alpha} \bmSigma_{\beta} \bmSigma_{\alpha})_{kj}\\ 
& \left. + \frac{n_{r\alpha} + n_{r\beta}}{n_{r\alpha}n_{r\beta}} \cdot(\bmSigma_{\alpha} \bmSigma_{\beta})_{kj}(\bmSigma_{\beta}\bmSigma_{\alpha})_{kj} \right].
\end{align*}
Here 
$(\bmSigma_{\alpha} \bmSigma_{\beta})_{kj}$ and $(\bmSigma_{\alpha} \bmSigma_{\beta} \bmSigma_{\alpha})_{kj}$ represent the $(k,j)$ elements of $\bmSigma_{\alpha} \bmSigma_{\beta}$ and $\bmSigma_{\alpha} \bmSigma_{\beta} \bmSigma_{\alpha}$, respectively, 
and $\bmSigma_{\alpha, \mathcal{N}(m) \mathcal{N}(m)}$ and $\bmSigma_{\beta, \mathcal{N}(m) \mathcal{N}(m)}$ represent sub-matrices of $\bmSigma_{\alpha}$ and $\bmSigma_{\beta}$, respectively, indexed by $\mathcal{N}(m)$.
Moreover, we have
\begin{align*}
\rho_{ab} \asymp \sqrt{\frac{n_{r\alpha} + n_{r\beta}}{n_{r\alpha}n_{r\beta}}  }\cdot \|\w_{ab}\|.
\end{align*}
\end{thm}

Next, we show the asymptotic results of the intermediate oracle LDSC estimator $\widetilde{\sigma}_{\alpha\beta}=(\bme_{ab}^T\bme_{ab})^{-1}\bme_{ab}^T\widehat{\w}_{ab}$. 
The following regularity conditions are similar to Conditions \ref{con_pn_relation_alpha} and \ref{con_condeff_lower_a} in Section \ref{sec3}. Specifically, Condition~\ref{con_pn_relation} imposes constraints on sample sizes $n_\alpha$ and $n_\beta$ of the two GWAS datasets, and Condition~\ref{con_condeff_lower} controls the lower bounds of the strength of genetic effects represented by $\|\bmalpha\|$ and $\|\bmbeta\|$.

\begin{condition}\label{con_pn_relation}
We assume that $p=O(n_{\alpha})$ and $p=O(n_{\beta})$.
\end{condition}
\begin{condition}\label{con_condeff_lower}
$\min(n_{\alpha}^{1/2}, n_{\beta}^{1/2}) \cdot \min(\|\bmalpha\|, 1) \cdot  \min( \|\bmbeta\|, 1) \to \infty$. 
\end{condition}

The following theorem establishes the asymptotic normality for $\widetilde{\sigma}_{\alpha\beta}$ and provides the explicit form of its variance.

\begin{thm}\label{thm_covariance_dep}
Under Conditions \ref{con1}, \ref{con_lmin_lower_2}, \ref{con_pn_relation}, and \ref{con_condeff_lower}, 
we have
\begin{equation*}
\frac{\widetilde{\sigma}_{\alpha\beta} - E\widetilde{\sigma}_{\alpha\beta}}{\zeta_{\alpha\beta} (\bme_{ab})} \overset{d}{\to} N(0,1)
\end{equation*}
{as $\min(n_{\alpha},n_{\beta},p)\to \infty$},
where $\zeta^2_{\alpha\beta}(\bme_{ab}) = \var(\widetilde{\sigma}_{\alpha\beta})$ and $\bm{D}_{\alpha\beta}=\text{diag}(\ell_{ab,1}, \ldots, \ell_{ab,p})$. In addition, we have
 \begin{align*}
 \zeta^2_{\alpha\beta} (\bme_{ab})=&  \frac{1}{(\bme_{ab}^T\bme_{ab})^{2} n_{\alpha}n_{\beta}} 
 \left[  (n_{\beta}+1) \bmalpha^T \bm{\Sigma}_{\alpha} \bmalpha \cdot \bmbeta^T \bm{\Sigma}_{\beta} \bm{D}_{\alpha\beta}  \bm{\Sigma}_{\alpha} \bm{D}_{\alpha\beta} \bm{\Sigma}_{\beta} \bmbeta  \right.\\
 & +\sigma^2_{\epsilon_{\beta}} \cdot \left\{
  \tr(\bm{D}_{\alpha\beta} \bm{\Sigma}_{\alpha} \bm{D}_{\alpha\beta} \bm{\Sigma}_{\beta} ) \cdot \bmalpha^T \bm{\Sigma}_{\alpha} \bmalpha 
  + (n_{\alpha}+1) \cdot \bmalpha^T \bm{\Sigma}_{\alpha}
 \bm{D}_{\alpha\beta} \bm{\Sigma}_{\beta} \bm{D}_{\alpha\beta} \bm{\Sigma}_{\alpha} \bmalpha \right\} \\
 & + \sigma^2_{\epsilon_{\alpha}}\cdot \left\{ \tr( \bm{D}_{\alpha\beta} \bm{\Sigma}_{\alpha} \bm{D}_{\alpha\beta} \bm{\Sigma}_{\beta}) \cdot \bmbeta^T \bm{\Sigma}_{\beta} \bmbeta 
 + (n_{\beta}+1) \cdot 
 \bmbeta^T \bm{\Sigma}_{\beta}
 \bm{D}_{\alpha\beta} \bm{\Sigma}_{\alpha} \bm{D}_{\alpha\beta} \bm{\Sigma}_{\beta} \bmbeta \right\}\\
 &  + \sigma^2_{\epsilon_{\beta}}  \sigma^2_{\epsilon_{\alpha}}\cdot
  \tr( \bm{D}_{\alpha\beta} \bm{\Sigma}_{\beta} \bm{D}_{\alpha\beta} \bm{\Sigma}_{\alpha}) 
  + (n_{\alpha}  + n_{\beta} +1) \cdot 
\left( \bmalpha^T \bm{\Sigma}_{\alpha} \bm{D}_{\alpha\beta} \bm{\Sigma}_{\beta} \bmbeta\right)^2 \\
& 
+ (n_{\alpha}+1) 
\cdot \bmbeta^T \bm{\Sigma}_{\beta} \bmbeta 
\cdot 
\bmalpha^T \bm{\Sigma}_{\alpha} \bm{D}_{\alpha\beta} \bm{\Sigma}_{\beta} \bm{D}_{\alpha\beta} \bm{\Sigma}_{\alpha} \bmalpha\\
& \left.
   +   \bmalpha^T \bm{\Sigma}_{\alpha} \bmalpha \cdot \bmbeta^T \bm{\Sigma}_{\beta} \bmbeta  \cdot \tr (\bm{\Sigma}_{\alpha}^{1/2} \bm{D}_{\alpha\beta}  \bm{\Sigma}_{\beta} \bm{D}_{\alpha\beta} \bm{\Sigma}_{\alpha}^{1/2})
 \right].
 \end{align*}
\end{thm}
%{\color{red}
\begin{remark}\label{thm_covariance_dep_remark}
In addition to cross-ancestry LD scores $\bme_{ab}$, Theorem \ref{thm_covariance_dep} holds for any generic vectors as long as their elements have absolute values bounded by constants from both above and below. This allows us to extend the theoretical results to estimated LD scores $\widehat{\bme}_{ab}$ under appropriate conditions.
\end{remark}
%}
One of the key challenges in proving Theorem~\ref{thm_covariance_dep} lies in addressing the dependence among the elements in $\widehat{\w}_{ab}$. Similar to the approach used in Theorem~\ref{thm_variance_dep}, we establish the asymptotic properties using the convergence in total variation for dependent random variables \citep{chatterjee2009fluctuations, dicker2014variance}. 
This allows us to effectively handle the high-dimensional dependence structure in GWAS summary data.
To establish the asymptotic normality for the LDSC estimator $\widehat{\sigma}_{\alpha\beta}$, 
we need the following additional conditions.

\begin{condition}\label{sigma_con}
We assume that
${\sigma_{\alpha\beta}=O(1)}$ and 
$\left|\bme_{ab}^T \bm{\varepsilon}_{ab}\right| = o(p^{1/2})$.
\end{condition}
\begin{condition}\label{con_coneff_lower}
We assume that 
$\|\bmalpha\| \|\bmbeta\|\gtrsim \sqrt{n_{\alpha} n_{\beta}}$.
\end{condition}
\begin{condition}\label{con_n_rate}
We assume that $\|\bmalpha\|\|\bmbeta\|=o(p^{3/2})$.
\end{condition}
\begin{condition}\label{con_margeff}
We assume that 
\begin{align*}
\frac{\|\w_{ab}\|^2}{p+n_{\alpha}+n_{\beta}} \left( \frac{1}{n_{r\alpha}} + \frac{1}{n_{r\beta}} \right) =o(1).
\end{align*}
\end{condition}
\begin{condition}\label{con_nr_rate_2}
We assume that {$\min(n_{r\alpha},n_{r\beta}) \gg \sqrt{p}$.} 
\end{condition}

{Similar to Condition~\ref{sigma_con_a},}
Condition~\ref{sigma_con} ensures the orthogonality between $\bme_{ab}$ and $\bm{\varepsilon}_{ab}$.
Conditions~\ref{con_coneff_lower} and \ref{con_n_rate} control the strength of the aggregated genetic signals.
Conditions~\ref{con_margeff} and \ref{con_nr_rate_2} impose constraints regarding the sample size of GWAS and reference panels.

\begin{thm}\label{thm_clt_cov}
Under Conditions \ref{con1}, 
\ref{con_lmin_lower},
and \ref{con_converge_rate}--\ref{con_nr_rate_2}, 
we have
\begin{align*}
\frac{\widehat{\sigma}_{\alpha\beta} - \sigma_{\alpha\beta}}{\zeta_{\alpha\beta} (\bme_{ab})}  \overset{d}{\to} N(0,1)
\end{align*}
and $\zeta_{\alpha\beta} (\bme_{ab})\to 0$
{as $\min(n_{\alpha},n_{r\alpha},n_{\beta},n_{r\beta},p)\to \infty$.}
\end{thm}

%{\bf \noindent .} 
\begin{proof}[Proof sketch of Theorem~\ref{thm_clt_cov}]
    We rewrite the standardized $\widehat{\sigma}_{\alpha\beta}$ as follows:
\begin{align*}
\frac{\widehat{\sigma}_{\alpha\beta} - \sigma_{\alpha\beta}}{\zeta_{\alpha\beta} (\bme_{ab})} 
=& \frac{(\widehat{\bme}_{ab}^T\widehat{\bme}_{ab})^{-1}\widehat{\bme}_{ab}^T\widehat{\w}_{ab} -  (\bme_{ab}^T\bme_{ab})^{-1}\bme_{ab}^T\w_{ab}}{\zeta_{\alpha\beta} (\bme_{ab})}  
+ \frac{(\bme_{ab}^T\bme_{ab})^{-1}\bme_{ab}^T \bm{\varepsilon}_{ab}}{\zeta_{\alpha\beta} (\bme_{ab})} \\
=& \frac{(\widehat{\bme}_{ab}^T\widehat{\bme}_{ab})^{-1}\widehat{\bme}_{ab}^T (\widehat{\w}_{ab} -\w_{ab})}{\zeta_{\alpha\beta} (\bme_{ab})}   
+ \frac{\{(\widehat{\bme}_{ab}^T\widehat{\bme}_{ab})^{-1}\widehat{\bme}_{ab}^T\ -  (\bme_{ab}^T\bme_{ab})^{-1}\bme_{ab}^T\} \w_{ab}}{\zeta_{\alpha\beta} (\bme_{ab})} \\ 
&+ \frac{(\bme_{ab}^T\bme_{ab})^{-1}\bme_{ab}^T \bm{\varepsilon}_{ab}}{\zeta_{\alpha\beta} (\bme_{ab})}. 
\end{align*}

Then it suffices to show that
\begin{align}\label{step1}
\frac{(\bme_{ab}^T\bme_{ab})^{-1}\bme_{ab}^T \bm{\varepsilon}_{ab}}{\zeta_{\alpha\beta} (\bme_{ab})} \to 0,
\end{align}
\begin{align}\label{step2}
\frac{\{(\widehat{\bme}_{ab}^T\widehat{\bme}_{ab})^{-1}\widehat{\bme}_{ab}^T\ -  (\bme_{ab}^T\bme_{ab})^{-1}\bme_{ab}^T\} \w_{ab}}{\zeta_{\alpha\beta} (\bme_{ab})} \overset{p}{\to} 0,
\end{align}
and 
\begin{align}\label{step3}
\frac{(\widehat{\bme}_{ab}^T\widehat{\bme}_{ab})^{-1}\widehat{\bme}_{ab}^T (\widehat{\w}_{ab} -\w_{ab})}{\zeta_{\alpha\beta} (\bme_{ab})}  \overset{d}{\to} N(0,1).
\end{align}

The lower and upper bounds of $\zeta_{\alpha\beta} (\bme_{ab})$ are both increasing functions of the  aggregated genetic effects $\|\bmalpha\|$ and $\|\bmbeta\|$.
Thus, to obtain  \eqref{step1},  we need an upper bound of $|\bme_{ab}^T \bm{\varepsilon}_{ab}|$ and a lower bound of  $\|\bmalpha\|\|\bmbeta\|$, which follow from Conditions \ref{sigma_con} and \ref{con_coneff_lower}, {\bxzzz respectively}.
Moreover, Condition \ref{con_n_rate} controls the upper bound of $\|\bmalpha\|\|\bmbeta\|$ to ensure $\zeta_{\alpha\beta} (\bme_{ab})\to 0$ and thus the convergence of $\widehat{\sigma}_{\alpha\beta}$ to $\sigma_{\alpha\beta}$.
{\bxzzz In addition,}
\eqref{step2} and \eqref{step3} can be derived from Theorems \ref{lTw_normality} and \ref{thm_covariance_dep}, respectively.
\end{proof}

In current within-ancestry applications, LDSC allows the data of the two traits to come from two different GWAS cohorts in the same population (e.g., {\bxzzz 1KG-EUR-like individuals}) and permits potential sample overlapping. This enables convenient estimation of genetic variance for two different traits using their publicly available GWAS summary statistics.
In our theoretical analysis, we demonstrate that LDSC can be further extended to cross-ancestry applications (e.g., {\bxzzz between 1KG-EUR-like and 1KG-EAS-like individuals}). 
Studying the differences and similarities in the genetic architecture underlying complex traits across global populations is of great interest \citep{peterson2019genome}. For example, it has been reported that the genetic correlation of schizophrenia is $0.98$ between {\bxzzz EAS and EUR genetic ancestry} cohorts \citep{lam2019comparative}, while the genetic correlation of major depressive disorder is only $0.41$ \citep{giannakopoulou2021genetic}.

Cross-ancestry LDSC also provides insights into the genetic heterogeneity across ancestries for different traits.
Curated genetic data resources, such as the
Pan-UKBB (\url{https://pan.ukbb.broadinstitute.org/}) and the global biobank meta-analysis initiative \citep{zhou2022global}, have been recently developed to promote multi-ancestry genetic research.
Using these data resources, it is possible to develop the infrastructure (such as performing QCs and calculating cross-ancestry LD scores) for cross-ancestry LDSC and share them with the community in the near future. In our {\bxzzz numerical} analysis, we have also tested cross-ancestry LDSC using simulated data.

Additionally, we establish the consistency of $\widehat{\sigma}_{\alpha\beta}$ 
{\fxintercept with weaker conditions in Section S2 of the Supplementary Material.}

% in Theorem~\ref{thm_convergence_cov} below. We find that Conditions \ref{con_coneff_lower}, \ref{con_margeff} and~\ref{con_nr_rate_2}
% can be relaxed {\bxzz to} the following conditions.
% \begin{condition}\label{con_margeff_2}
% We assume that 
% \begin{align*}
% \frac{\|\w_{ab}\|^2}{p^2} \left( \frac{1}{n_{r\alpha}} + \frac{1}{n_{r\beta}} \right) =o(1).
% \end{align*}
% \end{condition}
% \begin{condition}\label{con_nr_rate}
% There {\bxzzz exists} constants $c_1, c_2\in(0,1)$ such that $n_{r\alpha}^{c_1}\gg \log p$ and $n_{r\beta}^{c_2}\gg \log p$. That is, $\log p = o(n_{r\alpha}^{c_1})$ and $\log p = o(n_{r\beta}^{c_2})$.
% \end{condition}

% \begin{thm}\label{thm_convergence_cov}
% Under 
% Conditions 
% \ref{con1}, 
% \ref{con_lmin_lower},
% \ref{con_converge_rate}--\ref{sigma_con}, 
% \ref{con_n_rate}, 
% \ref{con_margeff_2},
% and
% \ref{con_nr_rate},
% we have
% \begin{align*}
% \widehat{\sigma}_{\alpha\beta} - \sigma_{\alpha\beta} \overset{p}{\to} 0,
% \end{align*}
% {as $\min(n_{\alpha},n_{r\alpha},n_{\beta},n_{r\beta},p)\to \infty$.}
% \end{thm}

In addition to LDSC, other methods have been proposed to estimate genetic variances and/or genetic covariances, particularly moment methods \citep{dicker2014variance,wang2022estimation,lu2017powerful,ma2019mahalanobis,schwartzman2019simple,zhou2017unified}. 
To help understand their connections, we provide a technical comparison between the LDSC and other existing methods in Section S1.1 of the Supplementary Material (Supplementary Table 1). Briefly, although these methods all can accept summary statistics as input, there are differences due to the nature of the proposed estimators and modeling frameworks. These differences include whether the methods and derived theoretical properties require consistent estimators of high-dimensional covariance/precision matrices, whether they model the randomness associated with using reference panels, and whether they consider bivariate and cross-ancestry applications. Notably, a key difference of the LDSC method is its use of LD scores, which take advantage of the block patterns of the LD structure and offer both good computational efficiency and theoretical properties.

Our framework is one of the first to explicitly model LD scores and their estimation from external reference panels, which is a common practical approach in genetic variance and covariance analyses. Previously established frameworks \citep{dicker2014variance,wang2022estimation} cannot be directly applied to study the LDSC estimators. We provide theoretical justifications for both univariate and bivariate applications, as well as cross-ancestry extensions, demonstrating that a relatively small sample size is sufficient for the reference panel in the practical estimation of LD scores.

%%%%%%%%%%%%%%%%%%%%%%%%%%%%%%%%%%%%%%%%%%%%%%%%%%%%%%%%%%%%%%%%%%%%
%%%%%%%%%%%%%%%%%%%%Numerical results%%%%%%%%%%%%%%%%%%%%%%%%%%%%%%%
%%%%%%%%%%%%%%%%%%%%%%%%%%%%%%%%%%%%%%%%%%%%%%%%%%%%%%%%%%%%%%%%%%%%
\section{Numerical results}\label{sec5}

%%%%%%%%%%%%%%%%%%%%%%%%%%%%%%%%%%%%%%%%%%%%%%%%%%%%%%%%%%%%%%%%%%%%
%%%%%%%%%%%%%%%%%%Simulated genotype data%%%%%%%%%%%%%%%%%%%%%%%%%%%
%%%%%%%%%%%%%%%%%%%%%%%%%%%%%%%%%%%%%%%%%%%%%%%%%%%%%%%%%%%%%%%%%%%%
\subsection{Simulated genotype data}\label{sec5.1}
We examine the numerical performance of the {\bxzz univariate and bivariate LDSC estimators} using simulated genotype data. We set $n_{\alpha}=n_{\beta}=2,\!000$ and $p=16,\!000$. 
The MAF $f$ of these genetic variants are independently generated from Uniform $[0.05, 0.45]$. 
Each entry of $\X_{\alpha_0}$ and $\X_{\beta_0}$ is independently sampled from $\{0,1,2\}$ with probabilities $\{(1-f)^2,2f(1-f),f^2\}$, respectively. 
The genome LD patterns are mimicked by constructing $\bmSigma_{\alpha}$ and $\bmSigma_{\beta}$ as block-diagonal matrices, each with $8$ large blocks.
The correlations among the $2,\!000$ genetic variants in each of the $8$ blocks are estimated from one region on chromosome $1$ (bp $40$-$50$m) with the {\bxzzz 1KG} reference panel \citep{10002015global}, and the genetic variants from different blocks are independent. We set $n_{r\alpha}=n_{r\alpha}=2,\!000$. 
We consider both within-ancestry and cross-ancestry genetic correlation estimations. 
In within-ancestry cases, we estimate $\bmSigma_{\alpha}=\bmSigma_{\beta}$ from European samples. In cross-ancestry cases, the $\bmSigma_{\alpha}$ and $ \bmSigma_{\beta}$ are estimated from European and East Asian subjects, respectively. 

The complex traits $\y_{\alpha}$ and $\y_{\beta}$ are simulated based on model~(\ref{equ2.1}). There is no sample overlap between the two traits. 
Let the SNP-based heritability \citep{yang2017concepts} of $\y_{\alpha}$ be $\h^2_{\alpha}=\bmalpha^T\bmalpha/(\bmalpha^T\bmalpha+\sigma^2_{\epsilon_\alpha})$ and that of $\y_{\beta}$ be $\h^2_{\beta}=\bmbeta^T\bmbeta/(\bmbeta^T\bmbeta+\sigma^2_{\epsilon_\beta})$.
In order to reflect low- and high-levels of heritability, 
we set $\h^2_{\alpha}=\h^2_{\beta}=$ $0.1$, $0.3$, $0.5$, or $0.8$. 
Similarly, the true genetic correlation $\varphi_{\alpha\beta}$ ranges from $0.1$, $0.3$, $0.5$, to $0.8$. 
The nonzero genetic effects in $\bmbeta$ and $\bmalpha$ are independently sampled from the normal distribution $N(0,1/p)$.
Let the number of nonzero elements in $\bmalpha$ and $\bmbeta$ be $m_{\alpha}$ and $m_{\beta}$, respectively. 
In order to evaluate sparse and dense genetic signals, we vary the proportion of genetic variants with nonzero effects $m_{\alpha}/p=m_{\beta}/p$ from $0.01$, $0.1$, $0.2$, $0.5$,  $0.8$, to $1$. 
We let $m_{\alpha\beta}/m_{\alpha}=m_{\alpha\beta}/m_{\beta}=0.5$ or $1$, where $m_{\alpha\beta}$ is the number of genetic variants that have nonzero effects on both  traits.
We evaluate the bivariate LDSC estimator $\widehat{\sigma}_{\alpha\beta}$
and  univariate LDSC estimators
$\widehat{\sigma}_{\alpha}^2$ and $\widehat{\sigma}_{\beta}^2$.
Each situation is replicated $500$ times. 

As expected, bivariate LDSC  estimator $\widehat{\sigma}_{\alpha\beta}$ is close to the underlying true values in both within-ancestry (Supplementary Figures 1--2) and cross-ancestry estimations (Supplementary Figures 3--4).  
Furthermore, we consistently find that cases with small sparsity (for example, $m_{\alpha}/p=m_{\beta}/p=0.01$) tend to have larger variations and more outliers. 
In addition, we observe that the results with $m_{\alpha\beta}/m_{\alpha}=m_{\alpha\beta}/m_{\beta}=0.5$ (Supplementary Figures 2 and 4) may have larger variations than those with $m_{\alpha\beta}/m_{\alpha}=m_{\alpha\beta}/m_{\beta}=1$ (Supplementary Figures 1 and 3). 
These results indicate that the variance of  $\widehat{\sigma}_{\alpha\beta}$ increases as $m_{\alpha}, m_{\beta}$, and $m_{\alpha\beta}$ decrease. LDSC has higher power for complex traits with greater polygenicity, where the number of common genetic variants influencing the phenotype is larger.
In Supplementary Figures 5 and 6, similar phenomena are observed for univariate LDSC estimators
$\widehat{\sigma}_{\alpha}^2$ and $\widehat{\sigma}_{\beta}^2$.
This numerical finding aligns with our theoretical results and discussions presented in Sections~\ref{sec2} and \ref{sec3}, where we need a large number of genetic variants with non-zero effects.
In summary, our results suggest that LDSC estimators can be applied to various combinations of heritability, genetic correlation, and genetic signal sparsity. Additionally, we demonstrate that the performance may be heavily influenced by genetic signal sparsity. LDSC may perform better for traits with a larger number of casual genetic variants. 

%%%%%%%%%%%%%%%%%%%%%%%%%%%%%%%%%%%%%%%%%%%%%%%%%%%%%%%%%%%%%%%%%%%%
%%%%%%%%%%%%%%%%%UK Biobank data analysis%%%%%%%%%%%%%%%%%%%%%%%%%%%
%%%%%%%%%%%%%%%%%%%%%%%%%%%%%%%%%%%%%%%%%%%%%%%%%%%%%%%%%%%%%%%%%%%%
\subsection{UK Biobank data analysis}\label{sec5.2}
We further evaluate the LDSC estimators using real genetic variant data from the UK Biobank study  \citep{bycroft2018uk}. 
We download the UKB imputed genotype data (version 3) and apply standard quality control steps, including 
1) excluding subjects with more than $10\%$ missing genotypes; 
2) removing genetic variants with MAF $\le 0.01$, with genotyping rate $\le 90\%$, or with Hardy-Weinberg test $p$-value $\le 1\times 10^{-7}$; and 
3) removing the variants with imputation INFO score less than $0.8$.
Following the original LDSC papers \citep{bulik2015ld,bulik2015atlas}, we further select the genetic variants in the HapMap3 reference panel \citep{international2010integrating} and remove the variants in the major histocompatibility complex region. After these quality control steps, there are $990,\!761$ genetic variants over $484,\!114$ subjects. 

To run LDSC with the developed software tool and {\bxzzz 1KG-EUR} LD scores at \url{https://github.com/bulik/ldsc}, we focus on within-ancestry analysis, in which we use the unrelated White British subjects ($n=366,\!335$). 
We consider three levels of sample size: $n_{\alpha}+n_{\beta}={\fxintercept 350,\!000}${\fxintercept, $50,\!000$, or $10,\!000$}. 
We randomly select these subjects and set $n_{\alpha}=n_{\beta}$ in a balanced-sample scenario and $n_{\alpha}=9\times n_{\beta}$ in an unbalanced-sample scenario. 
In each case, we consider three levels of sample overlap: 1) no overlap ($0\%$);
2) half of the $n_{\beta}$ subjects overlap with the $n_{\alpha}$ subjects  ($50\%$); 
and 3) all the $n_{\beta}$ subjects overlap with the $n_{\alpha}$ subjects ($100\%$). 
{\bxzz Let the genetic correlation between $\y_{\alpha}$ and $\y_{\beta}$ be $\varphi_{\alpha\beta}=g_{\alpha\beta}/(g_{\alpha}g_{\beta})$,} we set $\h_{\alpha}^2=\h_{\beta}^2=0.6$ and $\varphi_{\alpha\beta}=0.5$; or $\h_{\alpha}^2=\h_{\beta}^2=0.3$ and $\varphi_{\alpha\beta}=0.25$. The proportion of genetic variants with nonzero genetic effects $m_{\alpha}/p=m_{\beta}/p$ ranges from $0.001$, $0.01$, $0.1$, to $0.5$.  
We randomly select these  variants and their effects are independently generated from $N(0,1/p)$ using the GCTA \citep{yang2011gcta}. We obtain GWAS summary statistics by using fastGWA \citep{jiang2019resource}.  
We replicate each simulation setting $200$ times.

%%%%%%%%%%%%%%%%%%%%%Figure 2%%%%%%%%%%%%%%%%%%%%%%%%%%%%
\begin{figure}[t!]
\includegraphics[page=1,width=0.8\linewidth]{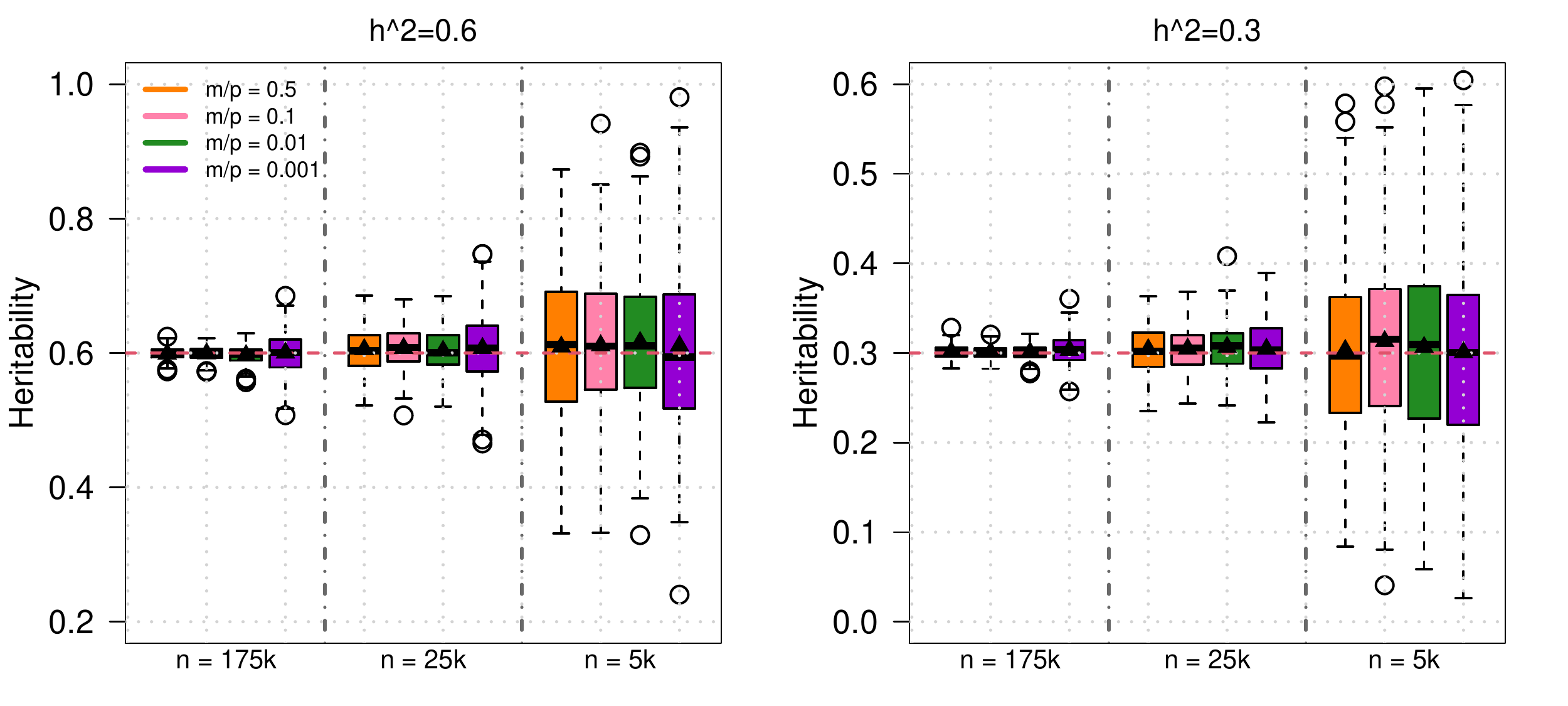}
  \caption{
  Univariate LDSC estimator across different sample sizes ($n_{\alpha}$), signal sparsity ($m/p$), and heritability level in the UK Biobank data simulation. We report the results of heritability, which is closely related to genetic variance. We set the heritability $h^2_{\alpha}=0.6$ and $0.3$ in the left and right panels, respectively. 
We simulate the data with $n_{\alpha}=175,000$, $25,000$, or $5,000$.  
The horizontal line represents the true heritability.
}
\label{fig2}
\end{figure}
%%%%%%%%%%%%%%%%%%%%%%%%%%%%%%%%%%%%%%%%%%%%%%%%%%%%%%%%

The results of {\bxzz univariate LDSC} are displayed in Figure~\ref{fig2}. 
We find  {\bxzz that the  univariate  LDSC estimators} are consistently  close to the underlying true values in different settings. The variance of estimators increases as the sample size and sparsity decrease. 
Similar patterns are demenstrated in Figure~\ref{fig3} and Supplementary Figures 7--8 for bivariate LDSC estimators.

There are some important observations. 
First, the variance of the LDSC estimators heavily depends on the sum of the two sample sizes ($n_{\alpha}+n_{\beta}$) and the genetic signal sparsity ($m_{\alpha}/p$ and $m_{\beta}/p$). When $n_{\alpha}+n_{\beta}={\fxintercept350,\!000}$, the variance of estimators in the $m_{\alpha}/p=m_{\beta}/p=0.001$ cases is much larger than those of denser genetic signals.  
As $n_{\alpha}+n_{\beta}$ decreases to $50,\!000$ or $10,\!000$, the variance level moves up and LDSC estimates become much noisier.

{\bxzzz We further assess the asymptotic normality of LDSC estimates, presented in Supplementary Figures 9--10, which include histograms and quantile-quantile plots across different sample sizes and sparsity levels. According to the Shapiro-Wilk test \citep{shapiro1965analysis}, the normality assumption for LDSC estimators generally holds, except in cases of small sample sizes with sparse genetic signals. These findings align with our theoretical analyses in Theorems \ref{thm_clt_var} and \ref{thm_clt_cov} for univariate and bivariate LDSC estimators, respectively, and also indicate that LDSC may require larger GWAS sample sizes for traits with sparser genetic signals to yield stable estimators. 
Additionally, we directly evaluate the asymptotic normality of the LD scores-related statistics studied in Lemmas \ref{lTl_normality_a} and \ref{lTl_normality}. 
The LD scores estimated from the 1KG reference panel provide numerical support for our theoretical analyses (Supplementary Figure 11).}

Second, LDSC has consistent performance across different levels of sample overlapping as shown in Figure~\ref{fig3} and Supplementary Figure 12. This is an attractive feature because in practice many traits may have partially or fully overlapped samples. 
In addition, the performance is also largely stable in balanced-sample and inbalanced-sample scenario. For example, $n_{\alpha}=n_{\beta}=175,\!000$ and $n_{\alpha}=315,\!000$ and $n_{\beta}=35,\!000$ have similar performance, indicating that LDSC can benefit from one trait with large sample size and allow the second trait has relatively small sample size. 

In summary, our GWAS data simulation results show that LDSC-based estimators have stable performance when there are sample overlaps. Overall, we need large sample sizes in GWAS, especially when the genetic signal is sparse. The authors of the LDSC software recommend that the minimum sample size to apply this method is about $5,\!000$ (\url{https://github.com/bulik/ldsc/wiki/FAQ}). 
{\bxzzz We discuss potential future directions to improve the performance for genetic variance and covariance estimations in the Discussion Section.}
%If the GWAS sample size for one of the two traits is very small (such as hundreds), the polygenic risk score-based genetic correlation estimators \citep{zhao2022genetic} may be applied.

%%%%%%%%%%%%%%%%%%Figure 3%%%%%%%%%%%%%%%%%%%%%%%%%%%%
\begin{figure}[t!]
\includegraphics[page=1,width=0.8\linewidth]{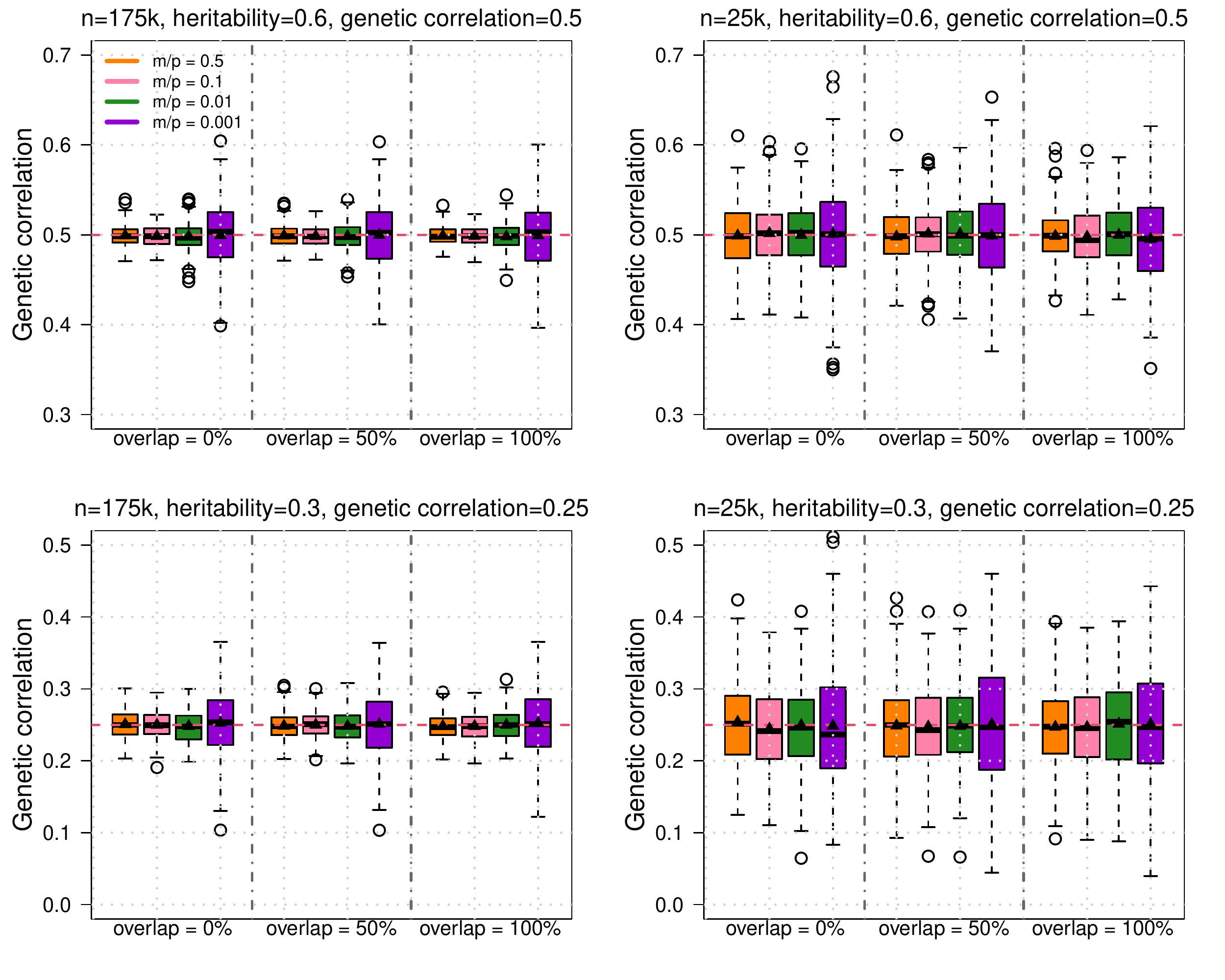}
  \caption{Bivariate LDSC estimator across different sample sizes ($n_{\alpha}$ and $n_{\beta}$), signal sparsity ($m/p$), and sample overlaps. 
  We report the genetic correlation, which is closely related to both genetic covariance and variance. 
  We set the heritability $h^2_{\alpha}=h^2_{\beta}=0.6$ and genetic correlation $\varphi_{\alpha\beta}=0.5$ in the top panels and 
 $h^2_{\alpha}=h^2_{\beta}=0.3$ and $\varphi_{\alpha\beta}=0.25$ in the bottom panels. 
We simulate the data with $n_{\alpha}=n_{\beta}=175,\!000$ or $25,\!000$. 
In each panel, we consider three cases of sample overlaps: 1) no sample overlap ($0\%$),
2) half sample overlap ($50\%$), and 3) all samples overlap ($100\%$). 
The horizontal line represents the true genetic correlations.
}
\label{fig3}
\end{figure}
%%%%%%%%%%%%%%%%%%%%%%%%%%%%%%%%%%%%%%%%%%%%%%%%%%%%%%%%%%%%%%%%%%%%
%%%%%%%%%%%%%%%%%HAPNEST data analysis%%%%%%%%%%%%%%%%%%%%%%%%%%%%%%
%%%%%%%%%%%%%%%%%%%%%%%%%%%%%%%%%%%%%%%%%%%%%%%%%%%%%%%%%%%%%%%%%%%%
{\bxzzz
\subsection{HAPNEST data analysis}\label{sec5.3}
To provide more insights into the extension of LDSC for cross-ancestry applications and the determination of LD blocks in such contexts, we conduct analyses using 1KG-EUR-like and 1KG-EAS-like genetic data generated by HAPNEST \citep{wharrie2023hapnest}.
%All of these LD scores were estimated from the 1000 Genomes \citep{10002015global}. 
The sample sizes are set at either $n_{\alpha}=n_{\beta}=168,\!000$ or $16,\!800$. We set $3\%$ of genetic variants to have nonzero genetic effects ($m_{\alpha}/p=m_{\beta}/p=0.03$), with a heritability of $h_{\alpha}^2=\h_{\beta}^2=0.5$, and genetic correlations set at $\varphi_{\alpha\beta}=0.5$ or $1$. 
To explore scenarios with different heritability values between two populations, we consider cases where heritability differs by $0.1$ ($h_{\alpha}^2=0.55$ and $\h_{\beta}^2=0.45$) 
and $0.4$ ($h_{\alpha}^2=0.7$ and $\h_{\beta}^2=0.3$). 
Other settings are the same as in Section~\ref{sec5.2}. 
Together, these result in $12$ different simulation setups and we estimate both heritability and genetic correlation in $200$ replications for each setup.

We evaluate different approaches for estimating LD scores from the 1KG project \citep{10002015global}: (i) Typical LD scores computed with the default window size (1 cM) using a pooled sample of 1KG-EUR and 1KG-EAS individuals and thus ignore the population differences ("Pooled subjects"); 
(ii) Cross-ancestry LD scores $\widehat{\bme}_{ab}$ estimated based on 1KG-EUR and 1KG-EAS subjects using 1MB-length blocks ("1MB-Window"); 
(iii) Cross-ancestry LD scores $\widehat{\bme}_{ab}$ estimated based on 1KG-EUR and 1KG-EAS subjects using 2MB-length blocks ("2MB-Window"); and 
(iv) Cross-ancestry LD scores $\widehat{\bme}_{ab}$ estimated based on 1KG-EUR and 1KG-EAS subjects using $1,368$ jointly inferred independent blocks with an average length of 2MB \citep{shi2020localizing} ("Independent"). 

We have several interesting findings (Supplementary Figures 13--14). 
First, both the "1MB-Window" and "2MB-Window" cross-ancestry LD score estimates provide unbiased results in all settings for both heritability and genetic correlation analyses.
These indicate that the window size-based LD score estimation strategy proposed by \cite{bulik2015ld} remains robust in cross-ancestry analyses. 
The inferred independent blocks \citep{shi2020localizing} result in a small bias, which becomes larger (yet still relatively minor) as the sample size increases.

In addition, consistent with expectations, we observe that pooled LD scores do not perform well in either heritability or genetic correlation analyses, regardless of whether the analysis uses meta-analysing GWAS summary statistics of two populations. 
For example, when the genetic correlation is set at $\varphi_{\alpha\beta}=1$ (bottom panels of Supplementary Figures 13--14), LD scores calculated from pooled samples lead to an underestimation of heritability in both EUR and EAS GWAS summary statistics. 
Using meta-analyzed GWAS summary statistics results in heritability estimates that fall between the two population-specific estimates, and they are still biased. 
When the genetic correlation is set at $\varphi_{\alpha\beta}=0.5$ (upper panels of Supplementary Figures 13--14), performing a meta-analysis may result in lower heritability estimates, which may reflect the partially overlapping genetic architecture between the traits in the two populations. 
Moreover, using pooled LD scores consistently leads to underestimated genetic correlation. These results suggest that overlooking population differences in LD score estimation and subsequent LDSC analyses may result in substantially biased results.

Therefore, the best practice for EUR-EAS LDSC genetic correlation analysis may be using fixed window sizes of either 1MB or 2MB to generate cross-ancestry LD scores $\widehat{\bme}_{ab}$. 
In Section S1.6 of the Supplementary Materials, we have also conducted comparisons of these methods across $34$ pairs of matched phenotypes from Biobank Japan \citep{sakaue2021cross} and UK Biobank \citep{bycroft2018uk}, which reveals similar patterns and suggest the robust performance of window size-based cross-ancestry LD scores in real data (Supplementary Table 2 and Supplementary Figure 15).}

%%%%%%%%%%%%%%%%%%%%%%%%%%%%%%%%%%%%%%%%%%%%%%%%%%%%%%%%%%%%%%%%%%%%
%%%%%%%%%%%%%%%%%%%%Discussion%%%%%%%%%%%%%%%%%%%%%%%%%%%%%%%%%%%%%%
%%%%%%%%%%%%%%%%%%%%%%%%%%%%%%%%%%%%%%%%%%%%%%%%%%%%%%%%%%%%%%%%%%%%
\section{Discussion}\label{sec6}
LDSC, which requires only GWAS summary data and the estimated LD scores from reference panels, has been widely applied in genetics and genomics.
In this paper, we propose a fixed-effect framework to investigate the theoretical properties of LDSC-based estimators. Our model setup incorporates a data integration approach to account for the utilization of GWAS and reference panel data from diverse data sources. 

Notably, our approach not only explicitly {\bxzzz models} the block-diagonal relatedness patterns in the estimated LD scores, but also addresses the high-dimensional dependence structure in GWAS summary data. Within this flexible modeling framework, we establish the asymptotic normality of both univariate and bivariate LDSC estimators. Furthermore, we extend our analysis to cross-ancestry applications, which are anticipated to play a significant role in future human genetic research \citep{zhou2022global}. 
{\bxzzz It is worth mentioning that the same trait may have different heritability across various cohorts and populations. Consequently, $h_{\alpha}^2$ and $h_{\beta}^2$ may differ in practice, especially in cross-ancestry analyses. Our theoretical framework for bivariate LDSC  permits the estimation of genetic correlation
between two different traits. 
Therefore, it is not necessary for $h_{\alpha}^2$ and $h_{\beta}^2$ to be identical to have an unbiased estimator of the genetic correlation.
%between two traits. 
Furthermore, 
given} the prevalence of reference panel-based methods in genetics and genomics \citep{hu2019statistical,xue2023causal}, our asymptotic normality and data integration frameworks hold broad applicability for studying the theoretical properties of various GWAS summary data-based methods.

It has been observed that, although easy to use, the LDSC estimators may only capture partial information about the LD matrix, leading to a loss in estimation precision. As a result, improving the performance of LDSC estimators has been an active research area in statistical genetics \citep{ning2020high,song2022leveraging,Smith2022.07.21.501001}.
{\bxzzz For example, HDL \citep{ning2020high} assumes genetic effects follow a high-dimensional normal distribution and uses maximum likelihood estimation to obtain genetic correlation parameters. This method leverages the entire LD matrix and the covariance matrix of GWAS summary statistics, not just LD scores as the LDSC does. 
In practical implementations, it relies on a banded eigen-decomposition procedure for computational feasibility.

Additionally, we consider the definitions of genetic variance and covariance as proposed by LDSC \citep{bulik2015ld} and GREML \citep{yang2011gcta}. 
While these definitions have been broadly used, they could be expanded in two directions to better model the intricate genetic architecture. First, we can include indirect genetic effects \citep{wang2022estimation}, and second, we can incorporate frequency-dependent relationships to reflect the impact of negative selection \citep{schoech2019quantification,momin2023method,speed2019sumher}. Briefly, the current definitions predominantly account for direct genetic effects from genetic variants, neglecting indirect genetic effects induced by the LD structure. Moreover, if complex traits' genetic architectures are influenced by negative selection, LDSC might yield biased estimates, especially for heritability estimation (Supplementary Figure 16). 
It is of great interest to extend our investigations to improve the performance of LDSC and understand its behavior under alternative model assumptions in future research. 
We provide more details in Section S1.5 of the Supplementary Material. 
}
%Finally, it is of great interest to extend our investigations to understand the behavior of other LDSC-based methods, such as the LD  eigenvalue regression \citep{song2022leveraging} and non-additive extensions of LDSC \citep{Smith2022.07.21.501001,palmer2023analysis}.

In our simulations, we also observe that LDSC may exhibit larger variance when the genetic signals are sparse. 
It would be interesting to {\bxzzz develop methods to improve the efficiency of genetic variance and covariance estimation} for traits with such genetic architecture. 
It is notable that directly selecting a subset of genetic variants with nonzero genetic effects may lead to even worse performance.
Considering the same setups in our UK Biobank simulation in Section~\ref{sec5.2}, we examine a second version of LDSC which only uses the genetic variants known to have a nonzero genetic effect. 
Supplementary {\bxzzz Figures 12 and 17} show that this version has an even larger variance in both heritability and genetic correlation estimations. This may be due to the fact that in the LDSC framework, the slope is estimated across the $p$ genetic variants. Naively selecting a subset of genetic variants may indeed reduce the "sample size" of the LDSC regression. 

{\bxzzz 
Therefore, for traits whose genetic architecture is estimated to be relatively sparse (which can be assessed, for example, using LD4M \citep{o2019extreme}, GENESIS \citep{zhang2018estimation}, or FMR \citep{o2021distribution}), LDSC-type regression-based estimators may not be efficient.  
For such traits, alternative ideas may be evaluated in future studies to improve efficiency. For example, methods developed with sparsity restrictions, such as FDEs \citep{guo2019optimal}, EstHer \citep{bonnet2018improving}, and CHIVE \citep{tony2020semisupervised}, could be extended to work with summary statistics and reference panels, potentially performing better for traits with highly sparse genetic signals.
These methods may remain effective for sparse traits even with low genome-wide heritability, as the smaller number of contributing genetic variants implies a stronger per-variant genetic signal. 
In addition, while methods based on summary data are convenient and enjoy the advantage of increasing GWAS sample sizes, they typically have lower statistical power compared to methods that work directly with individual-level data. 
Thus, for traits with low heritability that have not yet accumulated large GWAS sample sizes, methods based on individual-level data, such as GREML \citep{yang2011gcta,momin2023method}, may be preferred.
Furthermore, cross-trait polygenic risk score-based estimators \citep{power2015polygenic,zhao2022genetic} may be more effective for both sparse signal traits and low heritability traits. Many polygenic risk score methods, such as Lassosum \citep{mak2017polygenic} and PRSCS \citep{ge2019polygenic}, have been developed to handle flexible genetic architectures. These methods typically require individual-level data from one of the two traits, potentially resulting in smaller variance and better performance for low-heritability traits due to the use of individual-level data.
}
%If the GWAS sample size for one of the two traits is very small (such as hundreds), the polygenic risk score-based genetic correlation estimators \citep{zhao2022genetic} may be applied.
% and LDpred \citep{vilhjalmsson2015modeling}

Moreover, our simulation results indicate that there might be increased uncertainty in LDSC estimates when applied to a small subset of genetic variants, {\bxzzz such as} in stratified heritability enrichment analysis \citep{finucane2015partitioning,gazal2017linkage,finucane2018heritability}.
{\bxzzz  Intuitively, this is due to the fact that the number of genetic variants serve as the "sample size" in the LDSC framework.}
The conditions required for the asymptotic normality of LDSC estimators, as discussed in Sections~\ref{sec2} and \ref{sec3}, might be more challenging to satisfy when the number of genetic variants is substantially reduced, {\bxzzz leading to compromised reliability and accuracy of the estimates.}
These theoretical results are consistent with numerical observations reported in the genetic literature. Particularly, \cite{tashman2021significance} highlightes potential issues in stratified heritability enrichment testing when annotations are small. 
They note that stratified LDSC analyses are generally suitable for large annotations encompassing over $0.5\%$ of genome-wide genetic variants. 

We examine a series of existing annotations and find that they may have varying sizes and the majority has over $0.5\%$ of genetic variants, which may suggest that our developed theoretical results are applicable to many real functional annotations (Supplementary Figure 18). However, smaller annotations also exist, such as those more finely defined by brain cell subtypes \citep{hauberg2020common}, and are expected to become more prevalent with emerging cell-type specific data resources that provide fine-grained details in functional genomics.  Additionally, a typical approach in stratified heritability analysis is to compare the statistical significance across many annotations and report the top-ranking ones. This approach may be more suitable for functional annotations with comparable annotation sizes; otherwise, the size may confound the ranks as they serve as the "sample size" in LDSC. There is a need to develop new methods that work better for small annotations, such as gene-based and annotation-free approaches \citep{zhang2022polygenic}. Further details and suggestions are discussed in Section S1.4 of the Supplementary Material.
%Gene-level regression \citep{skene2018genetic} and per-cell scoring \citep{zhang2022polygenic} approaches might yield more reliable results for these cases. 

%alternative assumptions
%%%%%%%%%%%%%%%%%%%%%%%%%%%%%%%%%%%%%%%%%%%%%%%%%%%%%%%%%%%%%%%%%%%%
%%%%%%%%%%%%%%%%%%%%%Acknowledgement%%%%%%%%%%%%%%%%%%%%%%%%%%%%%%%%
%%%%%%%%%%%%%%%%%%%%%%%%%%%%%%%%%%%%%%%%%%%%%%%%%%%%%%%%%%%%%%%%%%%%
\section*{Acknowledgements}
The authors would like to thank the anonymous referees, the Associate Editor, and the Editor for their constructive comments, which significantly improved the quality of this paper.
We are also grateful to Hongyu Zhao, Jingyi Jessica Li, Jiwei Zhao, and Boran Gao for their insightful conversations and valuable suggestions.
%, which greatly contributed to strengthening this work.
The study has been partially supported by NSF Grant DMS 2210860 and start-up funds from Purdue Statistics Department. 
{\bxzzz Research reported in this publication is also supported by the National Institute of Mental Health under Award Number R01MH136055 and National Institute on Aging under Award Number RF1AG082938. The content is solely the responsibility of the authors and does not necessarily represent the official views of the National Institutes of Health.} 
This research has been conducted using the UK Biobank resource (application number 76139), subject to a data transfer agreement. We would like to thank the individuals and the research teams in the UK Biobank and Biobank Japan. We would like to thank the research computing groups at Purdue University and the Wharton School of the University of Pennsylvania for providing computational resources that have contributed to these research results. 

% %%%%%%%%%%%%%%%%%%%%%%%%%%%%%%%%%%%%%%%%%%%%%%
% %% Supplementary Material, including data   %%
% %% sets and code, should be provided in     %%
% %% {supplement} environment with title      %%
% %% and short description. It cannot be      %%
% %% available exclusively as external link.  %%
% %% All Supplementary Material must be       %%
% %% available to the reader on Project       %%
% %% Euclid with the published article.       %%
% %%%%%%%%%%%%%%%%%%%%%%%%%%%%%%%%%%%%%%%%%%%%%%
 \begin{supplement}
 {\bf Supplement to ``High-dimensional statistical inference for linkage disequilibrium score regression and its cross-ancestry extensions``.}
 Due to space constraints, {\bxzzz additional results and proofs} are deferred to the supplement.
% \stitle{Title of Supplement A}
% \sdescription{Short description of Supplement A.}
% \end{supplement}
% \begin{supplement}
% \stitle{Title of Supplement B}
% \sdescription{Short description of Supplement B.}
 \end{supplement}

\bibliographystyle{imsart-number} % Style BST file (imsart-number.bst or imsart-nameyear.bst)
\bibliography{sample.bib}       % Bibliography file (usually '*.bib')

%% or include bibliography directly:
% \begin{thebibliography}{4}
% %%
% \bibitem{r1}
% \textsc{Billingsley, P.} (1999). \textit{Convergence of
% Probability Measures}, 2nd ed.
% Wiley, New York.

% \bibitem{r2}
% \textsc{Bourbaki, N.}  (1966). \textit{General Topology}  \textbf{1}.
% Addison--Wesley, Reading, MA.

% \bibitem{r3}
% \textsc{Ethier, S. N.} and \textsc{Kurtz, T. G.} (1985).
% \textit{Markov Processes: Characterization and Convergence}.
% Wiley, New York.

% \bibitem{r4}
% \textsc{Prokhorov, Yu.} (1956).
% Convergence of random processes and limit theorems in probability
% theory. \textit{Theory  Probab.  Appl.}
% \textbf{1} 157--214.
% \end{thebibliography}

\end{document}